\documentclass[10pt,journal,twocolumn,romanappendices]{IEEEtran}
\usepackage{cite,graphicx,colordvi,psfrag}
\usepackage{amsmath,amssymb,nicefrac}
\usepackage{epstopdf}
\usepackage{epsfig}
\usepackage{calc,pstricks, pgf, xcolor}
\makeatletter
\newif\if@restonecol
\makeatother

\usepackage{algorithm2e}
\IEEEoverridecommandlockouts

\newcommand{\tpe}{{P}_{e,\text{pair}}}

\newcommand{\bTh}{_{\bar{1}}}
\newcommand{\bFo}{_{\bar{2}}}
\newcommand{\Th}{{\bar{1}}}
\newcommand{\Fo}{{\bar{2}}}

\newtheorem{property}{Property}
\newtheorem{theorem}{Theorem}

\newtheorem{remark}{Remark}

\newtheorem{definition}{Definition}
\newtheorem{lemma}{Lemma}

\newcommand{\BI}{ {\mathcal{I}}  }

\begin{document}

\title{Interference Alignment at Finite SNR for Time-Invariant Channels}
\author{Or~Ordentlich and
        Uri~Erez~\IEEEmembership{Member,~IEEE}

\IEEEcompsocitemizethanks{
\IEEEcompsocthanksitem
This work was supported in part by the Israel Science Foundation under grant 1557/10, the Binational Science Foundation under grant 2008455, a fellowship from The Yitzhak and Chaya Weinstein Research Institute for Signal Processing at Tel Aviv University, and the Feder Family Award.
\IEEEcompsocthanksitem
Or Ordentlich and Uri Erez are with the Department
of Electrical Engineering-Systems, Tel Aviv University, Ramat Aviv 69978, Israel.
\protect\\
E-mail: \{ordent,uri\}@eng.tau.ac.il.
}}

\maketitle
\begin{abstract}
An achievable rate region, based on lattice interference alignment, is derived for a class of time-invariant Gaussian interference channels with more than two users.
The result is established via a new coding theorem for the two-user Gaussian multiple-access channel where both users use a single linear code. The class of interference channels treated is such that all interference channel
gains are rational. For this class of interference channels, beyond recovering the known results on the degrees of freedom,
an explicit rate region is derived for finite signal-to-noise ratios, shedding light on the nature of previously established asymptotic results.
\end{abstract}
\begin{keywords}
Multiple access channel, linear codes, interference channel, interference alignment.
\end{keywords}

\section{Introduction}

An important open problem in network information theory is determining the capacity region of the interference channel. The interference channel is a communication model where multiple pairs of transmitters and receivers utilize the same communication medium. As a result, each user receives the output of a multiple-access channel, i.e., it suffers from interference from transmissions intended for other users.

An important special case of this channel model is the Gaussian interference channel, where each receiver sees a linear combination of its intended signal and the signals transmitted by the interfering users plus an additive white Gaussian noise (AWGN). For the case where only two users are sharing the same medium, i.e., the interference at each receiver is generated by only one user, the capacity region was characterized up to half a bit only recently~\cite{oneBit}. The achievability part utilizes the Han-Kobayashi~\cite{HanKobayashi} scheme which is shown to be nearly optimal in the two-user case. The results of ~\cite{oneBit} are rather disappointing in the sense that they imply that for a wide range of channel parameters, either treating the interference as noise, or alternating access to the medium (i.e., time sharing) between the two transmitter-receiver pairs, is a reasonable approach. In particular, time sharing yields the maximal degrees of freedom (DoF) afforded by the channel (i.e., one), where the number of DoF is defined as the ratio between the maximal possible sum rate and $\nicefrac{1}{2}\log({\text{SNR}} )$ in the limit where the SNR goes to infinity.

An interesting aspect of the interference channel is that the two-user case does not capture the quintessential features of the general ($K$-user)
interference channel, as has recently been demonstrated in the framework of linear Gaussian interference channels. In particular, while one may have suspected
that the channel would be interference limited, i.e., that time sharing would be optimal at high SNR, it has been demonstrated that this is not the case. Rather, it has been shown \cite{timeVarying,bobak,Etkin,Khandany} that the correct ``extension" of the two-user case is that in general, $K/2$ DoF are afforded by the $K$-user Gaussian interference channel.

\subsection{Related Work}
\label{subsec:relatedWork}
The works of \cite{timeVarying,bobak,Etkin,Khandany} have revealed that the Han-Kobayashi approach is inadequate for $K>2$, and a new approach, namely, interference alignment, was needed to achieve the DoF afforded by the (general) $K$-user interference channel.

The concept of interference alignment was introduced by Maddah-Ali et al. in~\cite{Xchannel} in the context of the multiple-input multiple-output (MIMO) X-channel, and then  by Cadambe and Jafar in ~\cite{timeVarying} for the $K$-user Gaussian interference channel with time-varying gains. The idea behind interference alignment in its original form, i.e., for time-varying channels, is essentially to split the signal space observed by each receiver into two orthogonal subspaces, where all the interference signals are confined to one, while the intended signal occupies the other. That is, if at \emph{each} receiver (simultaneously for all receivers) the interferences can be forced to lie in a subspace of dimension roughly half of that of the signal space, the rest of the signal space is free of interference and can be used for communicating the intended messages and hence roughly $K/2$ DoF can be utilized.

For the case of \emph{time or frequency varying channels}, under some mild ergodicity conditions,
this approach achieves the maximum possible DoF. Specifically, the interference alignment scheme of~\cite{timeVarying} achieves $K/2$ degrees of freedom while in~\cite{DoF_upperBound} it was shown that for a fully connected $K$-user interference channel, $K/2$ is also an upper bound on the number of DoF. It was later shown in~\cite{bobak} that for a large family of time or frequency varying channels,
each user can achieve about half of its interference-free ergodic capacity even at finite SNR.

The results of~\cite{timeVarying} and~\cite{bobak} are encouraging in the sense that they imply that the $K$-user Gaussian interference channel
is not inherently interference limited, a result that was quite unanticipated from the studies on the two-user case, which had been the focus of nearly all studies of Gaussian interference channels for nearly three decades.

The focus of this paper is the real \emph{time-invariant} (constant channel gains) $K$-user Gaussian interference channel, for which another from of interference alignment has proven to play a key role as well. In this case, it was shown in~\cite{Etkin, Khandany} that by taking the transmitted signal to belong to the integer lattice, it is possible to align the interference so that it remains confined to this lattice. As a result, the minimum distance of the received constellation at each receiver does not decrease with $K$, and when the SNR approaches infinity, each receiver can decode its intended signal with rate $\approx 1/4\log{\text{SNR}}$, yielding a total of $K/2$ DoF.
Thus, linear constellations, i.e., a PAM constellation in the one-dimensional case, play a key role in interference alignment for time-invariant channels.

Specifically, it was shown in~\cite{Etkin} that if at each receiver, the channel gains corresponding to the interferers are rational, whereas the direct channel gains corresponding to the intended signal are irrational, $K/2$ degrees of freedom are achievable. Even more interestingly, the authors of~\cite{Etkin} have shown that if the direct channel gains are rational as well, the degrees of freedom of the channel are strictly smaller than $K/2$. Later, the authors of~\cite{Khandany} proved that the DoF of the time-invariant interference channel are $K/2$ for almost all sets of channel gains.

As noted above, linear/lattice codes play a key role in coding for time-invariant Gaussian interference channels.
This feature of the $K$-user Gaussian interference channel is shared with a growing number of problems in
network information theory where lattice strategies have been shown to be an important ingredient. There are several examples where lattice strategies achieve better performance than the best known random coding strategies.
In particular, Philosof et al. introduced lattice interference alignment in the context of the doubly-dirty multiple-access channel \cite{Philosof}, i.e., to a Gaussian multiple-access channel with multiple interference signals, each known to a different transmitter.
Other network scenarios where lattices play a key role are the two-way (or multiple-way) relay problem ~\cite{Narayanan}
and in the compute \& forward approach to relay networks~\cite{compAndForIeee}.

Lattice interference alignment for the interference channel was first proposed by Bresler et al. in~\cite{BreslerAllerton,Bresler},
 where an approximate characterization of the capacity
 region for the many-to-one and one-to-many interference channels was derived. Lattice interference alignment was later utilized by Sridharan et al. in~\cite{VeryStrong} where a coding scheme where all users transmit points from the \emph{same} lattice was introduced. If at each receiver all the gains corresponding to the interferers are integers, the sum of the interferences is a point in the same lattice and thus the interference from $K-1$ users is confined to one lattice. Under very strong interference conditions, which are defined in~\cite{VeryStrong} and play the same role as the well-known very high interference condition~\cite{Carliel} for the two-user case, the decoder can first decode the sum of interferers while treating the desired lattice point as noise, then subtract the decoded interference, and finally decode the intended codeword. Later, in~\cite{LayeredSymmetric}, this scheme was combined with a layered coding scheme in order to show that lattice interference alignment can yield substantial gains and, in particular, achieve more that one DoF in some cases for a broader (but still quite limited) class of channels.


The works of~\cite{VeryStrong} and~\cite{LayeredSymmetric} allowed for important progress towards the understanding of interference alignment for finite SNR. Nonetheless, these results are limited since they essentially rely on using superposition with a judicious choice of power allocation such that a very strong interference condition holds, in conjunction with successive decoding. In the decoding procedure, a single layer is decoded at every step while the other layers are treated as noise. At each step of the successive decoding procedure, the decoder sees an equivalent point-to-point channel where lattice codes are used. The performance of lattice codes over AWGN point-to-point channels are well understood, and therefore the scheme of~\cite{VeryStrong} and~\cite{LayeredSymmetric}  can be analyzed with relative ease. For special classes of channel gains, it is possible to design a layered codebook that is simultaneously good for all receivers. For a wide range of channel parameters, however, such a layered scheme is not beneficial, as is also noted in~\cite{LayeredSymmetric}.

\subsection{Summary of Results}
\label{subsec:newContribution}

The main contribution of the present work is in providing a general framework for lattice interference alignment that is not
confined to successive decoding. A coding theorem is established for a multiple-access channel where all users use the same
linear codebook.

Specifically, if the interference is aligned to a lattice, but the very strong interference condition is not satisfied, the decoder can still perform \emph{joint decoding} of the interference codeword and the desired codeword. A major obstacle however arises when one attempts to jointly decode multiple codewords: The alignment of all interferers into one lattice point, which occurs simultaneously at all receivers, is only possible due to the fact that all users transmit lattice points from the \emph{same} lattice. Thus, if joint decoding is applied, each decoder sees a two-user Gaussian multiple-access channel (MAC) where both users use the \emph{same linear code}, for which, to the best of the authors knowledge, no previous results are known.

The fact that the number of DoF of many families of interference channels is $K/2$ implies that it is sufficient to transform the channel seen by each receiver into a two-user MAC in order to achieve optimal performance, i.e., half of the resources are dedicated to the interference and the other half to the intended signal. In light of this observation, along with the proven advantages of lattice codes in creating alignment between different users, it is important to study the achievable performance of the two-user Gaussian MAC where the same lattice code is used by all transmitters.

In this paper, we first address the question of finding an achievable symmetric rate for the Gaussian (modulo-additive) MAC with two users that are ``forced" to use the same linear code. We then employ this new ingredient in order to analyze an interference alignment scheme, suitable for a class of interference channels, which we refer to as the integer-interference channel, where all cross gains are integers (or rationals). The analysis is not asymptotic in the SNR.

While the proposed coding scheme does not require asymptotic conditions, we show that it is asymptotically optimal in a DoF sense, i.e., it achieves $K/2$ degrees of freedom for the integer-interference channel provided that the direct gains are irrational.
The achievable rate regions enables to shed light on the ``mystery'' around the effect of the direct channel gains being rational or irrational, which has to date only been understood for asymptotic high SNR conditions. In the proposed scheme, rational direct channel gains of the form $r/q$ limit the achievable symmetric rate to be smaller than $\log q$, which is not a serious limitation if $q$ is large and the SNR is moderate, but does indeed pose a limitation in the limit of very high SNR.

Moreover, previous results~\cite{Etkin,Khandany} state that the DoF of an interference channel with integer interference gains, are everywhere discontinuous in the direct channel gains. Such a result is quite displeasing and calls into question the applicability of interference alignment for time-invariant channels at non-asymptotic conditions, i.e., raises questions w.r.t. the robustness of lattice interference alignment. The results of this work demonstrate the behavior of the rate when the direct channel gains approach a given set of rational numbers. The (derived achievable) rate is continuous, as is to be expected, everywhere in the direct channel gains for any SNR, but the variation, i.e., sensitivity to the direct channel gain, increases with the SNR. 

While the presented scheme is only valid for channels where all the (non-direct) interference gains are integers (or rationals), we believe that the results are an important step towards the understanding of the feasibility of interference alignment for general time-invariant interference channels in the finite SNR regime.

The rest of this paper is organized as follows. In Section~\ref{sec:notations} some notations used throughout the paper are defined.
In Section~\ref{sec:MAC} an achievable symmetric rate is derived for a two-user Gaussian (modulo-additive) MAC where both users use the same linear code.
Section~\ref{sec:intAlignment} presents an interference alignment scheme for finite SNR. Section~\ref{sec:rationales} discusses the effect of the direct channel gains being rational vs. irrational on the performance of interference alignment. In Section~\ref{sec:nonInteger}, possible approaches for interference alignment when the interference gains are not restricted to be integers (or rational) are discussed. The paper concludes with Section~\ref{sec:conclusions}.

\section{Notational Conventions}
\label{sec:notations}

Throughout the paper we use the following natational conventions. Random variables are denoted by uppercase letters and their realizations by lowercase letters. For example $X$ is a random variable whereas $x$ is a specific value it can take. We use boldface letters to denote vectors, e.g., $\mathbf{x}$ denotes a vector with entries $x_i$.

A number of distinct modulo operations are used extensively throughout the paper.
The notation $x_{\bmod [a,b)}$ denotes reducing $x \in \mathbb{R}$ modulo the interval $[a,b)$.
That is, $x_{\bmod [a,b)}$ is equal to
$$x-m\cdot(b-a)$$
where $m\in\mathbb{Z}$ is the (unique) integer such that $$x-m\cdot(b-a)\in [a,b).$$
Similarly, $x_{\bmod p}$ where $x\in \mathbb{Z}_p$ is defined to equal
$$x-m\cdot p$$
where $m\in\mathbb{Z}$ is the unique integer such that
$x-m\cdot p\in \mathbb{Z}_p$.

If $\mathbf{x}$ is a vector, the notation $\mathbf{x}_{\bmod [a,b)}$ is understood to mean reducing each component of $\mathbf{x}$ modulo the interval $[a,b)$. We define the basic interval $\BI$ as $\left[-L/2,L/2\right)$ where $L=\sqrt{12}$. Reducing $x$  modulo the interval $\BI$ is denoted by $x^*$, i.e.,
$$x^*=[x]_{\bmod [-L/2,L/2)}.$$
The Euclidean norm of a vector $\mathbf{x}$ is denoted by $\left\|\mathbf{x}\right\|$. The notation $\lfloor x\rceil$ denotes rounding $x$ to the nearest integer. We denote the set of all prime numbers by $\mathcal{P}$. All logarithms in the paper are to the base $2$ and therefore all rates are expressed in bits per (real) channel use. All signals considered in this paper are real valued.

\section{Achievable Symmetric Rate for the Two-User Gaussian MAC with a Single Linear Codebook}
\label{sec:MAC}

\subsection{Problem statement}
\label{subsec:problemStatement}

We consider the modulo-additive MAC
\begin{align}
Y=\left[X_1+\gamma X_2+Z\right]^*,
\label{modMAC}
\end{align}
where $Z$ is an i.i.d. Gaussian noise with zero mean and variance $\mathbb{E}\left[Z^2\right]=1/\text{SNR}$.
We are interested in characterizing the achievable rate region for this channel where both users are forced to use the \emph{same linear} code, and where both users are subject to the power constraint
\begin{align}
\frac{1}{n}\mathbb{E}\left[ \left\|\mathbf{x}\right\|^2 \right] \leq 1.\nonumber
\end{align}
Note that a random variable uniformly distributed over $\BI$ has unit power.

An $(n,R)$ code for this model is defined by one encoding function
\begin{align}
f: \left\{1,\ldots,2^{nR}\right\}\rightarrow \BI^n\nonumber
\end{align}
and one decoding function
\begin{align}
g: 
 \BI^n \rightarrow \left\{1,\ldots,2^{nR}\right\}\times \left\{1,\ldots,2^{nR}\right\}\nonumber.
\end{align}
The linearity constraint on the encoding function $f$ is expressed by the condition that for
any $w_1,w_2 \in \left\{1,\ldots,2^{nR}\right\}$, there exists a $w_3\in\left\{1,\ldots,2^{nR}\right\}$
such that
\begin{align}
\left[f(w_1)+f(w_2)\right]^*=f(w_3).
 \label{linearity}
\end{align}
Specifically, user $1$ chooses a message $w_1\in\left\{1,\ldots,2^{nR}\right\}$ and transmits $\mathbf{x}_1=f(w_1)$, and user $2$ chooses a message $w_2\in\left\{1,\ldots,2^{nR}\right\}$ and transmits $\mathbf{x}_2=f(w_2)$. The decoder upon receiving
\begin{align}
\mathbf{y}=\left[\mathbf{x}_1+\gamma \mathbf{x}_2+\mathbf{z}\right]^*\nonumber,
\end{align}
generates estimates for the transmitted messages
\begin{align}
\left\{\hat{w}_1,\hat{w}_2\right\}=g(\mathbf{y})\nonumber.
\end{align}
The error probability for decoding the transmitted messages is denoted by
\begin{align}
P_e=\Pr\left(\{\hat{w}_1,\hat{w}_2\}\neq\{w_1,w_2\}\right)\nonumber.
\end{align}
We say that a symmetric rate $R$ is achievable if for any $\epsilon>0$ and $n$ large enough (depending on $\epsilon$), there exists an $(n,R)$ linear code such that $P_e<\epsilon$.

\subsection{Connection to the standard Gaussian MAC and previous results}
\label{subsec:prevResults}

The channel~(\ref{modMAC}) can be viewed as a degraded version of the Gaussian multiple access channel
\begin{align}
\tilde{Y}=X_1+\gamma X_2 + Z,
\label{nonModMAC}
\end{align}
as $Y$ can be obtained from $\tilde{Y}$ by the transformation
\begin{align}
Y=\tilde{Y}^*.\nonumber
\end{align}
It follows that the achievable symmetric rate for our channel model with the constraint that both users use the same linear code is upper bounded by that of the channel (\ref{nonModMAC}) where each user can use any codebook with rate $R$.

The capacity region of the Gaussian MAC (\ref{nonModMAC}) was characterized by Ahlswede \cite{MACcapacity} through the following equations
\begin{align}
R_1&<\frac{1}{2}\log\left(1+\text{SNR}\right)\nonumber\\
R_2&<\frac{1}{2}\log\left(1+\gamma^2\text{SNR}\right)\nonumber\\
R_1+R_2&<\frac{1}{2}\log\left(1+(1+\gamma^2)\text{SNR}\right).\nonumber
\end{align}
It follows that the symmetric capacity (i.e. the maximum achievable $R=R_1=R_2$) is given by
\begin{align}
C=\min\bigg\{\frac{1}{2}\log&\left(1+\text{SNR}\right),\frac{1}{2}\log\left(1+\gamma^2\text{SNR}\right),\nonumber\\
&\frac{1}{4}\log\left(1+(1+\gamma^2)\text{SNR}\right)\bigg\}.
\label{Rrandom}
\end{align}
The achievable part of the capacity theorem is proved using two \emph{different random} codebooks, whereas we restrict both users to use the \emph{same linear} (over the group $\BI$ with the addition operation) codebook. The main result of this section is the following theorem.

\subsection{Main result and discussion}
\label{sec3}

\begin{theorem}[{MAC with one linear code}]
For the setting described in Section~\ref{subsec:problemStatement} (a two-user Gaussian MAC channel where both users use the same linear code), the following symmetric rate is achievable
\begin{align}
&R<\max_{p\in\mathcal{P}'(\gamma)}\min\nonumber\\
& \bigg\{-\frac{1}{2}\log\left(\frac{1}{p^2}+\sqrt{\frac{2\pi/3}{\text{SNR}}}+\frac{1}{p}e^{-\frac{3\text{SNR}}{2p^2}\delta^2(p,\gamma)}+2e^{-\frac{3\text{SNR}}{8}}\right),\nonumber\\
&-\log\left(\frac{1}{p}+\sqrt{\frac{2\pi/3}{\delta^2(p,\gamma)\text{SNR}}}+2e^{-\frac{3\text{SNR}}{8}}\right)\bigg\},
\label{Capacity}
\end{align}
where
\begin{align}
\delta(p,\gamma)=\min_{l\in\mathbb{Z}_p\backslash\{0\}} l\cdot \left|\gamma-\frac{\lfloor l\gamma\rceil}{l}\right|,
\label{deltaDef}
\end{align}
and\footnote{Replacing the constraint ${p\in\mathcal{P}'(\gamma)}$ with the constraint $p\in\mathcal{P}$ results in a negligible change in the rate region described by~(\ref{Capacity}),(\ref{deltaDef}) for values of $\gamma$ that are not very ``close'' (w.r.t. the SNR) to $\pm 1/2$.}
\begin{align}
\mathcal{P}'(\gamma)=\left[p\in\mathcal{P}\bigg|e^{-\frac{3\text{SNR}}{2p^2}\left(\gamma_{\bmod [-\frac{1}{4},\frac{1}{4})}\right)^2}<1-2p \cdot e^{-\frac{3\text{SNR}}{8}}\right]
\label{thDef}
\end{align}
\label{capacityTheorem}
\end{theorem}
\textbf{Discussion:}
Inspecting the equations describing the achievable rate region of Theorem~\ref{capacityTheorem}, the role of the optimization parameter $p$, and the factor $\delta(p,\gamma)$ may seem at first strange. The factor $\delta(p,\gamma)$ is a measure of how accurately $\gamma$ can be approximated by a rational number with a denominator smaller than $p$. For example, if $\gamma$ is a rational number that can be written in the form $\gamma=r/q$,  then for any $p>q$ we have $\delta(p,\gamma)=0$. If this is the case, only values of $p\leq q$ yield non-trivial rates in (\ref{Capacity}), which in turn implies that for any value of SNR the rate of (\ref{Capacity}) is smaller than $\log q$.\footnote{This phenomenon is unique to the case where the codebook used by the two transmitters is linear. In the case where both transmitters use the same random codebook, it can be shown that (at least) the cutoff rate is identical to that of two different random codebooks (except for the singular case where $\gamma=1$ where it is impossible to distinguish between the users)}

\vspace{2mm}
\noindent \underline{\emph{Example}:}

In order to informally explain the rate saturation phenomenon (why $R<\log q$ for any SNR), we consider an example where $\gamma=1/3$. Assume that user $1$ transmits the codeword $\mathbf{x}_1=f(w_1)$ and user $2$ transmits the codeword $\mathbf{x}_2=f(w_2)$. We next observe that there are roughly $2^{nR}$ (up to a polynomial multiplicative term) ``competing'' pairs of codewords that can cause an error with a probability no smaller than $\left(1/3\right)^n$. Using the union bound as an approximation for the average error probability of the code, this implies that the achievable rate can be no greater than $\log 3$.

From the linearity property of the code (\ref{linearity}), there exists\footnote{Here we assume that the mapping $[3\mathbf{x}]^*$ gives a different result for every $\mathbf{x}\in\mathcal{C}$, as is the case for the ensemble of codebooks we consider in the sequel.} a codeword $\mathbf{x}_4=f(w_4)$ such that
\begin{align}
\mathbf{x}_2=\left[3\mathbf{x}_4\right]^*.\nonumber
\end{align}
Furthermore, for any codeword $\mathbf{x}_5=f(w_5)$, which is independent of $\mathbf{x}_1$ and $\mathbf{x}_2$, there exists a codeword $\mathbf{x}_3=f(w_3)$ such that
\begin{align}
\mathbf{x}_3=\left[3\mathbf{x}_5\right]^*.\nonumber
\end{align}
Now consider the pair of competing codewords
\begin{align}
{\mathbf{x}}\bTh=\left[\mathbf{x}_1+ \mathbf{x}_4+\mathbf{x}_5\right]^*\nonumber
\end{align}
and
\begin{align}
{\mathbf{x}}\bFo=\left[-{\mathbf{x}}_3\right]^*\nonumber
\end{align}
which exists by the linearity of the codebook. After passing through the channel (\ref{modMAC}), the ``distance'' between the transmitted pair of codewords and the ``competing'' pair of codewords is
\begin{align}
\mathbf{U}&=\left[\mathbf{x}_1+\frac{1}{3}\mathbf{x}_2-{\mathbf{x}}\bTh-\frac{1}{3}{\mathbf{x}}\bFo\right]^*\nonumber\\
&=\bigg[\mathbf{x}_1+\frac{1}{3}\mathbf{x}_2-\mathbf{x}_1-\mathbf{x}_4\nonumber-\mathbf{x}_5-\frac{1}{3}\left[-{\mathbf{x}}_3\right]^*\bigg]^*\nonumber\\
&=\bigg[\left(\frac{1}{3}\left[3\mathbf{x}_4\right]^*-\mathbf{x}_4\right)-\left(\mathbf{x}_5+\frac{1}{3}\left[-3\mathbf{x}_5\right]^*\right)\bigg]^*.
\label{exampleEq}
\end{align}
The terms
\begin{align}
\frac{1}{3}\left[3\mathbf{x}_4\right]^*-\mathbf{x}_4\nonumber
\end{align}
and
\begin{align}
\mathbf{x}_5+\frac{1}{3}\left[-3\mathbf{x}_5\right]^*\nonumber
 \end{align}
 can only take values in the set $\left\{-L/3,0,L/3\right\}$, which implies that $U$ can only take values in $\left\{-L/3,0,L/3\right\}$ as well.

If we further assume that each of the codewords in the codebook has a memoryless uniform distribution over the interval{\footnote{A uniform distribution over a uniform grid inside $\BI$ would have the same effect.} $\BI$, we can conclude that $U$ has a memoryless uniform distribution over $\left\{-L/3,0,L/3\right\}$. This means that after passing through the channel, the transmitted pair $\left\{\mathbf{X}_1,\mathbf{X}_2\right\}$ is equivalent to the competing pair $\left\{{\mathbf{X}}\bTh,{\mathbf{X}}\bFo\right\}$ with probability $\left(1/3\right)^n$. Since there are about $2^{nR}$ possible choices for $\mathbf{x}_3$, we see that rates above $R=\log 3$ would result in error with high probability.

While this conclusion relies on the union bound, which may not be tight for some cases, since almost every pair of two codewords from a randomly drawn linear codebook are statistically independent, the union bound is in fact not a bad approximation in the considered problem.

\vspace{2mm}
\noindent \underline{\emph{Comparison with random codebooks}:}

In order to better understand the performance of our coding scheme, we compare the maximum symmetric rate it achieves, which we refer to as $R_{\text{lin}}(\text{SNR})$, with that achieved by a coding scheme that utilizes two different random codebooks. We refer to the latter symmetric rate as $R_{\text{rand}}(\text{SNR})$ which is given by (\ref{Rrandom}). Define the normalized rate
\begin{align}
r_{\text{norm}}(\text{SNR})=\frac{R_{\text{lin}}(\text{SNR})}{R_{\text{rand}}(\text{SNR})}.\label{RnormEq}
\end{align}
Figure~\ref{RateFig} depicts $r_{\text{norm}}(\text{SNR})$ as a function of $\gamma\in[0,0.5)$ for a range of moderate to high values of SNR, specifically $\text{SNR}=20\text{dB}$, $30\text{dB}$ and $40\text{dB}$. Figure~\ref{RateFig2} depicts $r_{\text{norm}}(\text{SNR})$ as a function of $\gamma\in[0,0.5)$ for extremely high values of SNR, namely $\text{SNR}=100\text{dB}$, $110\text{dB}$ and $120\text{dB}$.

Figures~\ref{RateFig} and \ref{RateFig2} demonstrate the sensitivity of the rate to the channel gains. For a range of ``reasonable'' values of SNR, the rate changes rather smoothly with $\gamma$. For extremely high SNR, however, a slight change in the value of $\gamma$ may dramatically change the achievable rate.

The figures also suggest that for almost every value of $\gamma$, the normalized rate $r_{\text{norm}}(\text{SNR})$ approaches one as the SNR tends to infinity. Thus, the symmetric rate achieved when both users are using the same linear code scales with the SNR as $R_{\text{rand}}(\text{SNR})$ for asymptotic SNR conditions.

\begin{figure}[htb]
\includegraphics[width=1 \columnwidth]{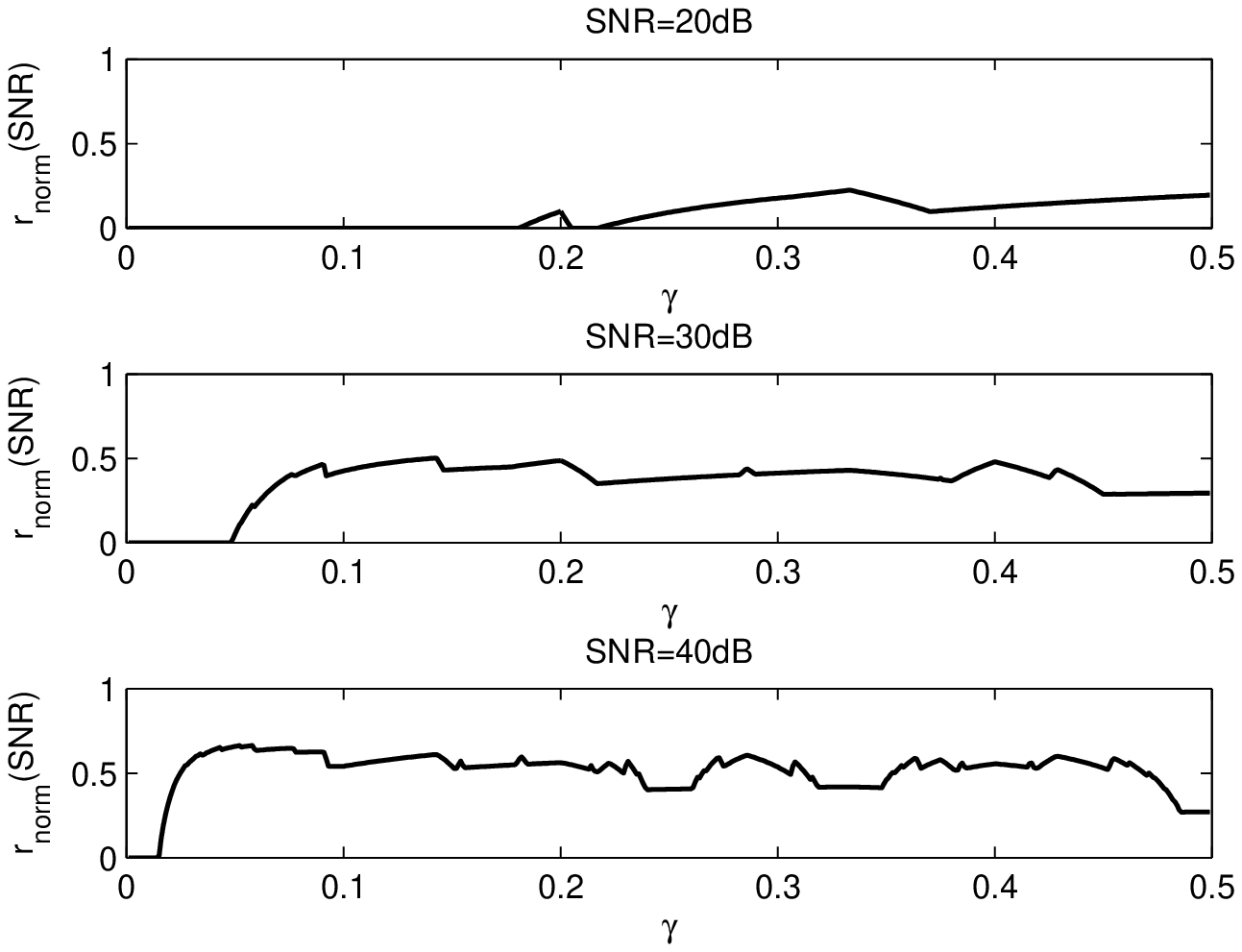}
\caption{$r_{\text{norm}}(\text{SNR})$ plotted as a function of $\gamma$ for $\text{SNR}=20\text{dB},30\text{dB},40\text{dB}$.}
\label{RateFig}
\end{figure}

\begin{figure}[htb]
\includegraphics[width=1 \columnwidth]{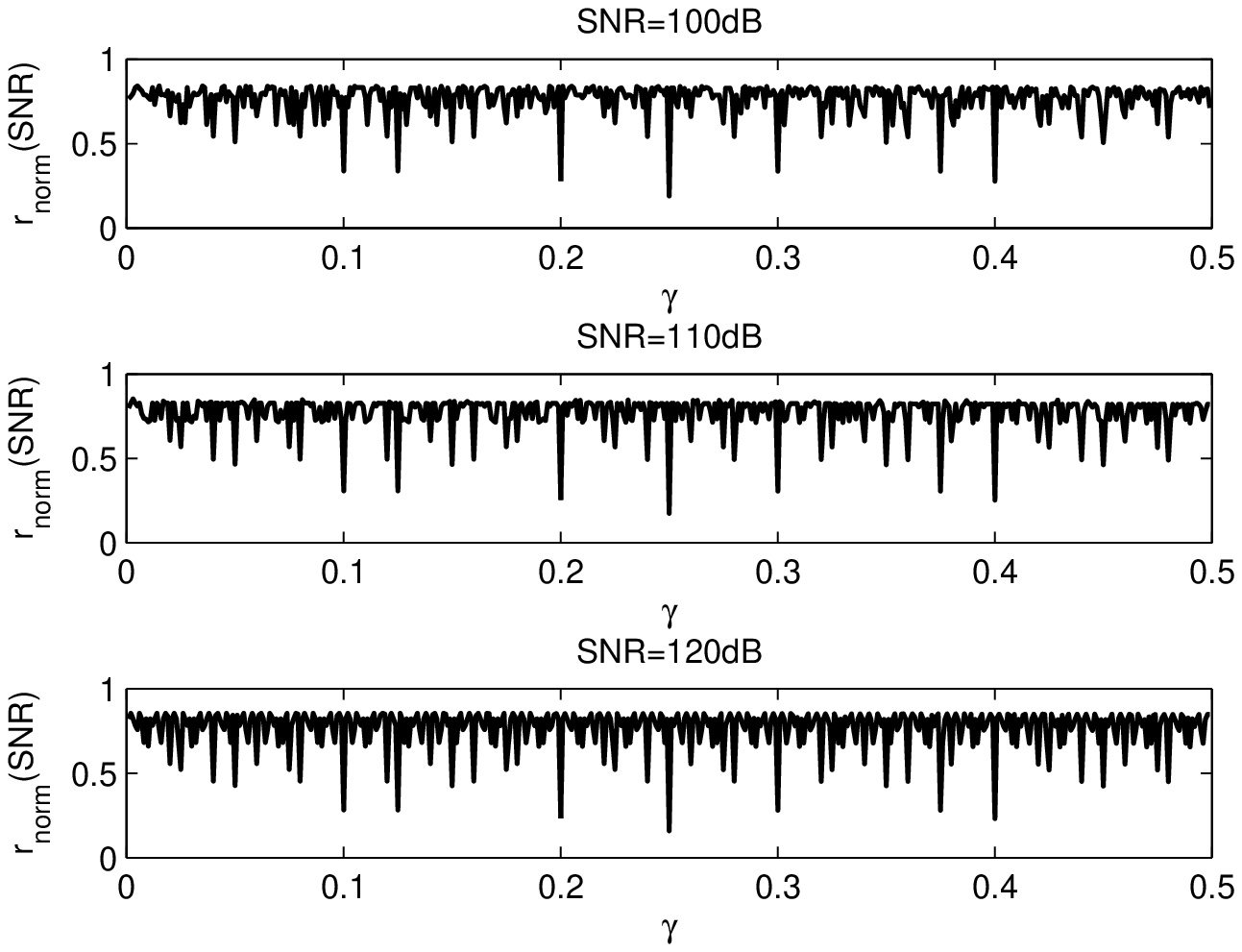}
\caption{$r_{\text{norm}}(\text{SNR})$ plotted as a function of $\gamma$ for $\text{SNR}=100\text{dB},110\text{dB},120\text{dB}$.}
\label{RateFig2}
\end{figure}

\vspace{2mm}

Note that Theorem~\ref{capacityTheorem} does not take into account shaping issues, since it uses a one-dimensional lattice as the coarse lattice. We have chosen not to pursue shaping in this paper in order to simplify the analysis. Moreover, the losses resulting from using the same linear code for both users outweigh the shaping loss, which can be upper bounded by a constant number of bits. Nevertheless, it is most likely that extending Theorem~\ref{capacityTheorem} to the case where both users use the same nested lattice codebook, with a ``good'' coarse lattice, can result in a higher symmetric rate for the two-user MAC.

\subsection{Proof of Theorem~\ref{capacityTheorem}}

We first describe the process of the code generation, the encoding and the decoding procedures, and then turn to analyze the error probability in decoding the transmitted messages.

\vspace{2mm}

\noindent \underline{\emph{Construction of linear codebook ensemble}:}
\vspace{1mm}

We begin by describing the generation process of the ensemble of linear codebooks considered, which is a variant of the well known Construction A (see e.g.~\cite{ConwaySloane}). A generating matrix $G$ of dimensions $k\times n$ is used, where all elements of $G$ are independently and randomly drawn according to the uniform distribution over the prime field $\mathbb{Z}_p$. We set
\begin{align}
k=n\frac{R}{\log p}.\nonumber
\end{align}
The set $\tilde{\mathcal{C}}$ is generated by multiplying all $k$-tuple vectors with elements from the field $\mathbb{Z}_p$ by the matrix $G$ (where all operations are over $\mathbb{Z}_p$)
\begin{align}
\tilde{\mathcal{C}}=\left\{\tilde{\mathbf{c}}=\mathbf{w}^T G \ | \ \mathbf{w}\in\mathbb{Z}_p^{k\times 1}\right\}.\nonumber
\end{align}
We refer to the vectors $\mathbf{w}$ as message vectors, and note that there are $2^{nR}$ such vectors, each corresponding to one of the possible messages.

Finally, the codebook $\mathcal{C}$ is generated from the set $\tilde{\mathcal{C}}$ by properly scaling and shifting it such that it meets the power constraint
\begin{align}
\mathcal{C}=\left[\frac{L}{p}\cdot \tilde{\mathcal{C}}\right]^*.\nonumber
\end{align}
The ensemble of codebooks created by the above procedure satisfies the following properties:
\begin{enumerate}
\item For any set of linearly independent message vectors, $\mathbf{w}_1,\mathbf{w}_2,\ldots,\mathbf{w}_l$, the corresponding codewords $\mathbf{X}_1,\mathbf{X}_2,\ldots,\mathbf{X}_l$ are statistically independent.
\item Each codeword $\mathbf{X}$ (except for the zero codeword) is memoryless:
\begin{align}
P(\mathbf{X})=\prod_{t=1}^{n} P(X_t).\nonumber
\end{align}
\item For any $\mathbf{w}\neq\mathbf{0}$ the corresponding codeword is uniformly distributed over the constellation
    \begin{align}
    \Lambda=\frac{L}{p}\left[\mathbb{Z}^n_p\right]_{\bmod [-\frac{p}{2},\frac{p}{2})}.\nonumber
    \end{align}
\item Each codeword $\mathbf{X}$ in the ensemble satisfies the power constraint
\begin{align}
\frac{1}{n}\mathbb{E}\left[ \|\mathbf{X}\|^2 \right] \leq 1.\nonumber
\end{align}
\item Each codebook in the ensemble satisfies the linearity constraint as defined in (\ref{linearity}).
\end{enumerate}

\vspace{2mm}

\noindent \underline{\emph{Encoding}:}
\vspace{1mm}

Suppose a codebook from the above ensemble, which is completely characterized by the matrix $G$, has been chosen.  User $i$ uniformly draws a message vector $\mathbf{w}_i$, and transmits
\begin{align}
\mathbf{x}_i=\left[\frac{L}{p}\mathbf{w}_i^T G\right]^*.
\label{codewordGeneration}
\end{align}
The channel output is thus
\begin{align}
\mathbf{y}=\left[\mathbf{x}_1+\gamma\mathbf{x}_2+\mathbf{z}\right]^*.\nonumber
\end{align}

\vspace{2mm}

\noindent \underline{\emph{Decoding}:}
\vspace{1mm}

Given the encoding matrix $G$ and the channel output $\mathbf{y}$, the decoder searches the pair of codewords $\left\{\mathbf{x}_i,\mathbf{x}_j\right\}$ for which
\begin{align}
\mathbf{\psi}(i,j)=\left[\mathbf{x}_i+\gamma\mathbf{x}_j\right]^*
\label{psiDef}
\end{align}
is closest to $\mathbf{y}$ in the following sense
\begin{align}
\left(\hat{i},\hat{j}\right)=\arg\min_{i,j}\left(\sum_{t=1}^n\left(\left[y_t-\psi_t(i,j)\right]^*\right)^2\right).
\label{decoderEq}
\end{align}
If there is more than one pair of indices satisfying (\ref{decoderEq}), an error is declared.

The decoder only searches over the pairs of codewords corresponding to message vectors $\left\{\mathbf{w}_i,\mathbf{w}_j\right\}$ that are linearly independent (over $\mathbb{Z}_p$).
This constraint on the decoder facilitates the analysis of the error probability since it means that when $G$ is assumed to be random, the decoder only searches over the pairs of codewords $\left\{\mathbf{X}_i,\mathbf{X}_j\right\}$ which are statistically independent. The above constraint implies that if the users had chosen message vectors $\left\{\mathbf{w}_i,\mathbf{w}_j\right\}$ which are linearly dependent, an error event occurs. For the rest of the analysis we assume that indeed the chosen message vectors are linearly independent, and as a consequence $\left\{\mathbf{X}_i,\mathbf{X}_j\right\}$ are statistically independent when $G$ is assumed to be random. We account for the probability of the error event that occurs when this is not the case, in the final step of the proof.

We note that the decision rule (\ref{decoderEq}) is an approximation of the maximum-likelihood decoder which searches for a pair of codewords $\left\{\mathbf{x}_i,\mathbf{x}_j\right\}$ that satisfies
\begin{align}
\Pr\left(\mathbf{y}|\mathbf{x}_i,\mathbf{x}_j\right)>\Pr\left(\mathbf{y}|\mathbf{x}_m,\mathbf{x}_l\right) \ \ \forall (m,l)\neq (i,j).\nonumber
\end{align}
The suboptimal decoding algorithm we use simplifies the analysis, but does not incur significant losses w.r.t. the optimal maximum-likelihood decoder (for the channel~(\ref{modMAC})).

\vspace{2mm}

\noindent \underline{\emph{Analysis of error probability}:}
\vspace{1mm}

We analyze the average error probability over the ensemble of codebooks described above, i.e., we assume the generating matrix $G$ is random, and average over all possible realizations of $G$.

Assume that the message vectors $\left\{\mathbf{w}_1,\mathbf{w}_2\right\}$ were chosen by users $1$ and $2$ respectively, such that codeword $\mathbf{X}_1$ was transmitted by user $1$, and $\mathbf{X}_2$ by user $2$. We first analyze the pairwise error probability, i.e., the probability of the decoder preferring a (different) specific pair of message vectors $\left\{\mathbf{w}\bTh,\mathbf{w}\bFo\right\}$, corresponding to the pair of codewords $\left\{\mathbf{X}\bTh,\mathbf{X}\bFo\right\}$, over the transmitted pair.

As we recall, due to the linear structure of the codebook, linear dependencies within the set of chosen and ``competing'' message vectors $\left\{\mathbf{w}_1,\mathbf{w}_2,\mathbf{w}\bTh,\mathbf{w}\bFo\right\}$ result in statistical dependencies within the set of transmitted and ``competing'' codewords $\left\{\mathbf{X}_1,\mathbf{X}_2,\mathbf{X}\bTh,\mathbf{X}\bFo\right\}$.
We are interested in the average pairwise error probability associated with each pair of ``competing'' message vectors $\left\{\mathbf{w}\bTh,\mathbf{w}\bFo\right\}$. Thus, the average pairwise error probability has to be analyzed w.r.t. each one of the possible statistical dependencies. We develop upper bounds on the average pairwise error probability associated with each type of statistical dependency, and then invoke the union bound in order to establish an upper bound on $\mathbb{E}[P_e]$, the average probability of the decoder not deciding on the correct pair of transmitted codewords. Using this bound, an achievable rate region is obtained.

Denote the pairwise error probability from the pair of message vectors $\left\{\mathbf{w}_1,\mathbf{w}_2\right\}$ to the pair $\left\{\mathbf{w}\bTh,\mathbf{w}\bFo\right\}$, for a given codebook in the ensemble, by $\tpe$, and the average pairwise error probability over the ensemble by $\mathbb{E}\left[\tpe\right]$.

We begin by deriving a general expression that upper bounds the average pairwise error probability $\mathbb{E}\left[\tpe\right]$ and then evaluate it for each of the possible statistical dependencies within the set $\left\{\mathbf{X}_1,\mathbf{X}_2,\mathbf{X}\bTh,\mathbf{X}\bFo\right\}$.

The decoder makes an error to the pair $\left\{\mathbf{X}\bTh,\mathbf{X}\bFo\right\}$ only if
\begin{align}
\sum_{t=1}^n\big(&\left[Y_t-\Psi_t(1,2)\right]^*\big)^2\geq\sum_{t=1}^n\big(\left[Y_t-\Psi_t(\Th,\Fo)\right]^*\big)^2,
\label{decoderError}
\end{align}
where $\Psi(i,j)$ is defined in~(\ref{psiDef}).
The condition in~(\ref{decoderError}) is equivalent to
\begin{align}
\sum_{t=1}^n&\left(Z_t^*\right)^2\geq\sum_{t=1}^n\left(\left[Z_t+\Psi_t(1,2)-\Psi_t(\Th,\Fo)\right]^*\right)^2.
\label{errorDeriv}
\end{align}
Define the \emph{pairwise difference} random variable
\begin{align}
U_t&=\left[\Psi_t(1,2)-\Psi_t(\Th,\Fo)\right]^*\nonumber\\
&=\left[X_{1,t}+\gamma X_{2,t}-X_{\Th,t}-\gamma X_{\Fo,t}\right]^*.
\label{Udefinition}
\end{align}
and vector
\begin{align}
\mathbf{U}=\left[U_1 \ U_2 \ \ldots \ U_n\right].
\label{Uvec_definition}
\end{align}
Note that the distribution of the pairwise difference vector $\mathbf{U}$ encapsulates the statistical dependencies in the set of codewords $\left\{\mathbf{X}_1,\mathbf{X}_2,\mathbf{X}\bTh,\mathbf{X}\bFo\right\}$. We first express our upper bounds on the average pairwise error probability as a function of the random vector $\mathbf{U}$, and only then account for the fact that the statistics of $\mathbf{U}$ vary with the different types of statistical dependencies between the transmitted and the ``competing'' pairs of codewords.

Substituting (\ref{Udefinition}) into (\ref{errorDeriv}), we have that an error occurs only if
\begin{align}
\sum_{t=1}^n(Z_t^*)^2\geq\sum_{t=1}^n\left(\left[Z_t+U_t\right]^*\right)^2.
\label{randomErrorEvent}
\end{align}
Given a specific codebook from the ensemble was chosen, $\mathbf{U}$ is deterministic, and~(\ref{randomErrorEvent}) implies
\begin{align}
\tpe&= \Pr\left(\left\|\mathbf{Z}^*\right\|^2
\geq\left\|\left[\mathbf{Z}+\mathbf{u}\right]^*\right\|^2\right)\nonumber\\
&= \Pr\left(\left\|\mathbf{Z}^*\right\|^2\geq\min_{\mathbf{v}\in L\mathbb{Z}^n}
\left\|\mathbf{Z}^*+\mathbf{u}+\mathbf{v}\right\|^2\right).
\label{foldedNoise1}
\end{align}
Let $\mathbb{T}^n=\left\{-1,0,1\right\}^n$. Since every coordinate of the vectors $\mathbf{u}$ and $\mathbf{z}^*$ has an absolute value smaller than $L/2$, the value of $\mathbf{v}\in L\mathbb{Z}^n$ that minimizes the expression $\left\|\mathbf{z}^*+\mathbf{u}+\mathbf{v}\right\|^2$ cannot have an absolute value greater than $L$ in any component, and it suffices to limit the search for it to $L\mathbb{T}^n$. Hence (\ref{foldedNoise1}) simplifies to
\begin{align}
\tpe=\Pr\left(\left\|\mathbf{Z}^*\right\|^2\geq\min_{\mathbf{v}\in L\mathbb{T}^n}
\left\|\mathbf{Z}^*+\mathbf{u}+\mathbf{v}\right\|^2\right).
\label{foldedNoise}
\end{align}
We now state a simple lemma that enables us to replace the folded Gaussian noise $\mathbf{Z}^*$ in (\ref{foldedNoise}) with a simple Gaussian noise $\mathbf{Z}$.

\vspace{2mm}
\begin{lemma}
\label{moduloLemma}

For $\mathbf{z}\in\mathbb{R}^n$, $\mathbf{u}\in\BI^n$ and the events
\begin{align}
E_1=\left\{\left\|\mathbf{z}^*\right\|^2\geq\min_{\mathbf{v}\in L\mathbb{T}^n}
\left\|\mathbf{z}^*+\mathbf{u}+\mathbf{v}\right\|^2\right\},
\label{event1}
\end{align}
and
\begin{align}
E_2=\left\{\left\|\mathbf{z}\right\|^2\geq\min_{\mathbf{\tilde{v}}\in L\mathbb{T}^n}
\left\|\mathbf{z}+\mathbf{u}+\mathbf{\tilde{v}}\right\|^2\right\},
\label{event2}
\end{align}
the following relation holds
\begin{align}
E_1\subseteq E_2.\nonumber
\end{align}
\end{lemma}
\vspace{2mm}

\begin{proof}
See Appendix~\ref{proofOfLemmaMod}.
\end{proof}
\vspace{2mm}

The next lemma provides an upper bound on $\mathbb{E}\left[\tpe\right]$, the average pairwise error probability over the ensemble, that depends only on $\mathbb{E}\left[\exp\left\{-\frac{\text{SNR}}{8}U^2\right\}\right]$.

\vspace{2mm}

\begin{lemma}
The average pairwise error probability over the ensemble is upper bounded by
\begin{align}
\mathbb{E}\left[\tpe\right]\leq\Omega^n,\nonumber
\end{align}
where
\begin{align}
\Omega=\mathbb{E}\left[\exp\left\{-\frac{\text{SNR}}{8}U^2\right\}\right]+2\exp\left\{-\frac{3\text{SNR}}{8}\right\}.
\label{OmegaEq}
\end{align}
\end{lemma}
\vspace{2mm}
\begin{proof}
Using Lemma~\ref{moduloLemma}, we have
\begin{align}
\tpe&=\Pr\left(\left\|\mathbf{Z}^*\right\|^2\geq\min_{\mathbf{v}\in L\mathbb{T}^n}
\left\|\mathbf{Z}^*+\mathbf{u}+\mathbf{v}\right\|^2\right)\nonumber\\
&\leq\Pr\left(\left\|\mathbf{Z}\right\|^2\geq\min_{\mathbf{\tilde{v}}\in L\mathbb{T}^n}
\left\|\mathbf{Z}+\mathbf{u}+\mathbf{\tilde{v}}\right\|^2\right)\nonumber\\
&=\Pr\left(\bigcup_{\mathbf{\tilde{v}}\in L\mathbb{T}^n}\left(\left\|\mathbf{Z}\right\|^2\geq\
\left\|\mathbf{Z}+\mathbf{u}+\mathbf{\tilde{v}}\right\|^2\right)\right).
\label{nonFoldedGauss}
\end{align}
Using the union bound,~(\ref{nonFoldedGauss}) can be further bounded by
\begin{align}
\tpe&\leq\sum_{\mathbf{\tilde{v}}\in L\mathbb{T}^n}\Pr\left(\left\|\mathbf{Z}\right\|^2\geq\
\left\|\mathbf{Z}+\mathbf{u}+\mathbf{\tilde{v}}\right\|^2\right)\nonumber\\
&=\sum_{\mathbf{\tilde{v}}\in L\mathbb{T}^n}\Pr\left(-(\mathbf{u}+\mathbf{\tilde{v}})^T\mathbf{Z}\geq\
\frac{1}{2}\left\|\mathbf{u}+\mathbf{\tilde{v}}\right\|^2\right).
\label{gaussUB}
\end{align}
Since $\mathbf{Z}$ is a vector of i.i.d. Gaussian components with zero mean and variance $1/\text{SNR}$, the random variable $-(\mathbf{u}+\mathbf{\tilde{v}})^T\mathbf{Z}$ is Gaussian with zero mean and variance $\|\mathbf{u}+\mathbf{\tilde{v}}\|^2/\text{SNR}$. Using the notation
\begin{align}
Q(\tau)=\int_{\tau}^{\infty}\frac{1}{\sqrt{2\pi}}\exp\left\{-\frac{1}{2}\tau^2\right\} d\tau,\nonumber
\end{align}
and recalling that
\begin{align}
Q(\tau)\leq\exp\left\{-\frac{1}{2}\tau^2\right\},\nonumber
\end{align}
(\ref{gaussUB}) becomes
\begin{align}
\tpe&\leq\sum_{\mathbf{\tilde{v}}\in L\mathbb{T}^n}Q\left(\frac{\sqrt{\text{SNR}}}{2}\|\mathbf{u}+\mathbf{\tilde{v}}\|\right)\nonumber\\
&\leq\sum_{\mathbf{\tilde{v}}\in L\mathbb{T}^n}\exp\left\{-\frac{\text{SNR}}{8}\|\mathbf{u}+\mathbf{\tilde{v}}\|^2\right\}.\nonumber
\end{align}
In order to find the average (over the ensemble) pairwise error probability, we need to average $\tpe$ according to the distribution of $\mathbf{U}$
\begin{align}
\mathbb{E}\left[\tpe\right]\leq\sum_{\mathbf{\tilde{v}}\in L\mathbb{T}^n}\mathbb{E}&\bigg[\exp\left\{-\frac{\text{SNR}}{8}\|\mathbf{U}+\mathbf{\tilde{v}}\|^2\right\}\bigg].
\label{PeExp1}
\end{align}
Since the code generation is memoryless, $\mathbf{U}$ is also memoryless, and (\ref{PeExp1}) can be rewritten as
\begin{align}
\mathbb{E}[\tpe]\leq\sum_{\mathbf{\tilde{v}}\in L\mathbb{T}^n}\prod_{t=1}^n
\mathbb{E}\left[\exp\left\{-\frac{\text{SNR}}{8}(U_t+\tilde{v}_t)^2\right\}\right].
\label{PeExp2}
\end{align}
Equation~(\ref{PeExp2}) can be further simplified by replacing the order between the sum and the product
\begin{align}
&\mathbb{E}[\tpe]\leq\prod_{t=1}^n\sum_{{\tilde{v}}\in L\mathbb{T}}
\mathbb{E}\left[\exp\left\{-\frac{\text{SNR}}{8}(U_t+\tilde{v})^2\right\}\right]\nonumber\\
&=\left(\sum_{{\tilde{v}}\in L\mathbb{T}}\mathbb{E}\left[\exp\left\{-\frac{\text{SNR}}{8}(U+\tilde{v})^2\right\}\right]\right)^n\label{identU}\\
&=\bigg(\mathbb{E}\left[\exp\left\{-\frac{\text{SNR}}{8}(U-L)^2\right\}\right]+\mathbb{E}\left[\exp\left\{-\frac{\text{SNR}}{8}U^2\right\}\right]\nonumber\\
&+\mathbb{E}\left[\exp\left\{-\frac{\text{SNR}}{8}(U+L)^2\right\}\right]\bigg)^n\nonumber\\
&\leq \left(\mathbb{E}\left[\exp\left\{-\frac{\text{SNR}}{8}U^2\right\}\right]+2\exp\left\{-\frac{\text{SNR}}{8}\frac{L^2}{4}\right\}\right)^n\label{boundedU}\\
&=\left(\mathbb{E}\left[\exp\left\{-\frac{\text{SNR}}{8}U^2\right\}\right]+2\exp\left\{-\frac{3\text{SNR}}{8}\right\}\right)^n\nonumber\\
&=\Omega^n\nonumber
\end{align}
where~(\ref{identU}) follows from the fact that the random variables $\left\{U_t\right\}_{t=1}^n$ are identically distributed, and~(\ref{boundedU}) is true since $|U|\leq L/2$, and thus $(U+L)^2\geq L^2/4$ as well as $(U-L)^2\geq L^2/4$.
\end{proof}
\vspace{2mm}
In order to obtain an explicit upper bound on the average (over the ensemble) pairwise error probability, we are left with the task of calculating $\Omega$, or equivalently calculating $\mathbb{E}\left[\exp\left\{-\frac{\text{SNR}}{8}U^2\right\}\right]$.

We recall that $\mathbf{U}$ is a deterministic function of the pair of transmitted codewords $\left\{\mathbf{X}_1,\mathbf{X}_2\right\}$ and the pair of ``competing'' codewords $\left\{\mathbf{X}\bTh,\mathbf{X}\bFo\right\}$, where each one of the codewords is generated as specified in~(\ref{codewordGeneration}). The statistical dependencies within the set of codewords $\left\{\mathbf{X}_1,\mathbf{X}_2,\mathbf{X}\bTh,\mathbf{X}\bFo\right\}$ correspond to the linear dependencies within the set of message vectors $\left\{\mathbf{w}_1,\mathbf{w}_2,\mathbf{w}\bTh,\mathbf{w}\bFo\right\}$. Since we assumed the message vectors $\left\{\mathbf{w}_1,\mathbf{w}_2\right\}$ are linearly independent,\footnote{The message vectors $\left\{\mathbf{w}\bTh,\mathbf{w}\bFo\right\}$ are also linearly independent, as the decoder only searches over the pairs of linearly independent message vectors.} there are only four possible cases of linear dependencies within the set $\left\{\mathbf{w}_1,\mathbf{w}_2,\mathbf{w}\bTh,\mathbf{w}\bFo\right\}$:

\emph{Case A}: The four vectors $\left\{\mathbf{w}_1,\mathbf{w}_2,\mathbf{w}\bTh,\mathbf{w}\bFo\right\}$ are linearly independent.

\emph{Case B}: The vectors $\left\{\mathbf{w}_1,\mathbf{w}_2,\mathbf{w}\bTh\right\}$ are linearly independent and $\mathbf{w}\bFo$ is a linear combination of them.

\emph{Case C}: The vectors $\left\{\mathbf{w}_1,\mathbf{w}_2,\mathbf{w}\bFo\right\}$ are linearly independent and $\mathbf{w}\bTh$ is a linear combination of them.

\emph{Case D}: The vectors $\left\{\mathbf{w}_1,\mathbf{w}_2\right\}$ are linearly independent and both $\mathbf{w}\bTh$ and $\mathbf{w}\bFo$ are linear combination of them.

\vspace{2mm}

Each case of statistical dependencies induces a different distribution on $U$. Thus for the calculation of $\mathbb{E}\left[\exp\left\{-\frac{\text{SNR}}{8}U^2\right\}\right]$, each case should be considered separately. To that end, we now give upper bounds on $\Omega$ for the four different possible cases.  The derivations of these bounds are given in Appendix~\ref{dependencies}.
\vspace{2mm}

\emph{Case A}:  The codewords $\{\mathbf{X}_1,\mathbf{X}_2,\mathbf{X}\bTh,\mathbf{X}\bFo\}$ are all statistically independent. Given $\left\{\mathbf{w}_1,\mathbf{w}_2\right\}$, there are less than $2^{2nR}$ pairs of competing message vectors $\left\{\mathbf{w}\bTh,\mathbf{w}\bFo\right\}$ that incur this kind of statistical dependency. Denote by $\Omega_A$ the value of $\Omega$ associated with case $A$. We have
\begin{align}
\Omega_A<\frac{1}{p^2}+\sqrt{\frac{2\pi/3}{\text{SNR}}}+\frac{1}{p}e^{-\frac{3\text{SNR}}{2p^2}\delta^2(p,\gamma)}+2e^{-\frac{3\text{SNR}}{8}},
\label{PeA}
\end{align}
where
\begin{align}
\delta(p,\gamma)=\min_{l\in\mathbb{Z}_p\backslash\{0\}} l\cdot \left|\gamma-\frac{\lfloor l\gamma\rceil}{l}\right|.
\label{delta}
\end{align}
\vspace{2mm}

\emph{Case B}: The codewords $\{\mathbf{X}_1,\mathbf{X}_2,\mathbf{X}\bTh\}$ are statistically independent and
\begin{align}
\mathbf{X}\bFo=\left[a\mathbf{X}_1+b\mathbf{X}_2+c\mathbf{X}\bTh\right]^*,\nonumber
\end{align}
where $a$, $b$, and $c$ can take any value in $\mathbb{Z}_p$. Given $\left\{\mathbf{w}_1,\mathbf{w}_2\right\}$, there are no more than $p^3 2^{nR}$ pairs of competing message vectors $\left\{\mathbf{w}\bTh,\mathbf{w}\bFo\right\}$ that incur this kind of statistical dependency. Denote by $\Omega_B$ the value of $\Omega$ associated with case $B$. We have
\begin{align}
\Omega_B<\frac{1}{p}+\sqrt{\frac{2\pi/3}{\delta^2(p,\gamma)\text{SNR}}}+2e^{-\frac{3\text{SNR}}{8}},
\label{PeB}
\end{align}
where $\delta(p,\gamma)$ is as in (\ref{delta}).

\vspace{2mm}

\emph{Case C}: The codewords $\{\mathbf{X}_1,\mathbf{X}_2,\mathbf{X}\bFo\}$ are statistically independent and
\begin{align}
\mathbf{X}\bTh=\left[a\mathbf{X}_1+b\mathbf{X}_2+c\mathbf{X}\bFo\right]^*,\nonumber
\end{align}
 where $a$, $b$, and $c$ can take any value in $\mathbb{Z}_p$. Given $\left\{\mathbf{w}_1,\mathbf{w}_2\right\}$, there are no more than $p^3 2^{nR}$ pairs of competing message vectors $\left\{\mathbf{w}\bTh,\mathbf{w}\bFo\right\}$ that incur this kind of statistical dependency. Denote by $\Omega_C$ the value of $\Omega$ associated with case $C$. We have
\begin{align}
\Omega_C<\frac{1}{p}+\sqrt{\frac{2\pi/3}{\delta^2(p,\gamma)\text{SNR}}}+2e^{-\frac{3\text{SNR}}{8}},
\label{PeC}
\end{align}
where $\delta(p,\gamma)$ is as in (\ref{delta}). Note that although the bounds (\ref{PeB}) and (\ref{PeC}) are identical, cases B and C
are not identical (i.e., there is no symmetry) since the two codewords $\mathbf{X}\bTh$ and $\mathbf{X}\bFo$ play a different role in the pairwise difference vector $\mathbf{U}$, as
$\mathbf{X}\bFo$ is multiplied by $\gamma$ while $\mathbf{X}\bTh$ is not.

\vspace{2mm}

\emph{Case D}: The codewords $\{\mathbf{X}_1,\mathbf{X}_2\}$ are statistically independent, whereas
\begin{align}
\mathbf{X}\bTh=\left[a\mathbf{X}_1+b\mathbf{X}_2\right]^*,\nonumber
\end{align}
and
\begin{align}
\mathbf{X}\bFo=\left[c\mathbf{X}_1+d\mathbf{X}_2\right]^*,\nonumber
\end{align}
where $a$, $b$, $c$ and $d$ can take any value in $\mathbb{Z}_p$, except for $a=1,b=0,c=0,d=1$ (in which case $\mathbf{X}\bTh=\mathbf{X}_1$ and $\mathbf{X}\bFo=\mathbf{X}_2$). Given $\left\{\mathbf{w}_1,\mathbf{w}_2\right\}$, there are no more than $p^4$ pairs of competing message vectors $\left\{\mathbf{w}\bTh,\mathbf{w}\bFo\right\}$ that incur this kind of statistical dependency. Denote by $\Omega_D$ the value of $\Omega$ associated with case $D$. We have
\begin{align}
\Omega_D<\max\bigg\{&\frac{1}{p}+\sqrt{\frac{2\pi/3}{\delta^2(p,\gamma)\text{SNR}}}+2e^{-\frac{3\text{SNR}}{8}},\nonumber\\
&\frac{p-1}{p}+\frac{1}{p}e^{-\frac{3\text{SNR}}{2p^2}\left(\gamma_{\bmod [-\frac{1}{4},\frac{1}{4})}\right)^2}+2e^{-\frac{3\text{SNR}}{8}}\bigg\}.
\label{PeD}
\end{align}
where $\delta(p,\gamma)$ is given in (\ref{delta}).

\vspace{2mm}

We can now establish the theorem. Denote by $\mathbb{E}\left[{P}_{e,\text{pair},i}\right], \ i=A,B,C,D$, the average error probability associated with each case of statistical dependencies. We recall that the decoder in our scheme only searches over the pairs of codewords corresponding to message vectors $\left\{\mathbf{w}_i,\mathbf{w}_j
\right\}$ that are linearly independent. Thus, an error event occurs if the message vectors $\left\{\mathbf{w}_1,\mathbf{w}_2
\right\}$ chosen by the users are linearly dependant. Denote by ${P}_{e,E}$ the probability of this event (which is independent of the codebook).
By basic combinatorics
\begin{align}
{P}_{e,E}=(p+1)\cdot 2^{-nR}-p\cdot2^{-2nR}<2p\cdot2^{-nR}.\nonumber
\end{align}
Using the union bound, the average error probability over the ensemble can be upper bounded by
\begin{align}
\mathbb{E}\left[P_e\right]&\leq 2^{2nR}\cdot \mathbb{E}[{P}_{e,\text{pair},A}]+p^3 \cdot 2^{nR} \cdot \mathbb{E}[{P}_{e,\text{pair},B}]\nonumber\\
&+p^3 \cdot 2^{nR} \cdot \mathbb{E}[{P}_{e,\text{pair},C}]+p^4 \cdot \mathbb{E}[{P}_{e,\text{pair},D}]+{P}_{e,E}\nonumber\\
&\leq 2^{2n\left(R+\frac{1}{2}\log\Omega_A\right)}\nonumber\\
&+2^{n\left(R+3\frac{\log p}{n}+\log\Omega_B\right)}\nonumber\\
&+2^{n\left(R+3\frac{\log p}{n}+\log\Omega_C\right)}\nonumber\\
&+2^{n\left(\frac{4\log p}{n}+\log\Omega_D\right)}\nonumber\\
&+2^{n\left(-R+\frac{\log 2p}{n}\right)}.\nonumber
\end{align}
Holding $p$ constant and taking $n$ to infinity we see that the average error probability goes to zero if
\begin{align}
R<&-\frac{1}{2}\log\Omega_A,\label{R1a}\\
R<&-\log\Omega_B,\label{R2a}\\
R<&-\log\Omega_C,\label{R3a}\\
0>&\log \Omega_D.\label{R4a}
\end{align}
The conditions (\ref{R1a}), (\ref{R2a}) and~(\ref{R3a}) imply that the rate should be taken to satisfy
\begin{align}
R<&\min \nonumber\\ &\bigg\{-\frac{1}{2}\log\left(\frac{1}{p^2}+\sqrt{\frac{2\pi/3}{\text{SNR}}}+\frac{1}{p}e^{-\frac{3\text{SNR}}{2p^2}\delta^2(p,\gamma)}+2e^{-\frac{3\text{SNR}}{8}}\right)\nonumber\\
&-\log\left(\frac{1}{p}+\sqrt{\frac{2\pi/3}{\delta^2(p,\gamma)\text{SNR}}}+2e^{-\frac{3\text{SNR}}{8}}\right)\bigg\},
\label{RateExp}
\end{align}
whereas condition (\ref{R4a}) implies
\begin{align}
\frac{1}{p}+\sqrt{\frac{2\pi/3}{\delta^2(p,\gamma)\text{SNR}}}+2e^{-\frac{3\text{SNR}}{8}}<1,
\label{CondD4a}
\end{align}
and
\begin{align}
\frac{p-1}{p}+\frac{1}{p}e^{-\frac{3\text{SNR}}{2p^2}\left(\gamma_{\bmod [-\frac{1}{4},\frac{1}{4})}\right)^2}+2e^{-\frac{3\text{SNR}}{8}}<1.
\label{CondD4b}
\end{align}
Condition~(\ref{CondD4a}) is satisfied for any positive rate, since it is contained in~(\ref{RateExp}). Condition~(\ref{CondD4a})
is equivalent to
\begin{align}
e^{-\frac{3\text{SNR}}{2p^2}\left(\gamma_{\bmod [-\frac{1}{4},\frac{1}{4})}\right)^2}<1-2pe^{-\frac{3\text{SNR}}{8}}.
\label{pCond}
\end{align}
Since any prime value of $p$ that satisfies (\ref{pCond}) is valid, we can maximize (\ref{RateExp}) over all prime values of $p$ satisfying (\ref{pCond}), i.e. over all values in $\mathcal{P}'(\gamma)$ as defined in~(\ref{thDef}), which yields (\ref{Capacity}). Finally, since there must be at least one codebook in the ensemble with a smaller (or equal) error probability than the average over the ensemble, the theorem is proved.

\section{Application to interference alignment}
\label{sec:intAlignment}

In the previous section we found an achievable symmetric rate for the Gaussian modulo-additive MAC channel where both users use the same linear codebook. The motivation for developing such a coding scheme is to enable \emph{lattice interference alignment}.

Assume a receiver observes a linear combination of codewords transmitted by several users (corrupted by noise) and is interested in decoding only one of the codewords, namely the received signal is
\begin{align}
\mathbf{y}=h_1\mathbf{x}_1+\sum_{k=2}^K h_k \mathbf{x}_k+\mathbf{z},\nonumber
\end{align}
where $\left\{h_k\right\}_{k=1}^K$ are the channel gains, $\mathbf{x}_1$ is the desired codeword, $\left\{\mathbf{x}_k\right\}_{k=2}^K$ are the interfering codewords, and $\mathbf{z}$ is a vector of i.i.d. Gaussian noise.

One approach is to treat all the interfering codewords as noise. This approach would not be effective when the total power of the interference is on the order of that of the desired codeword (or stronger). Another possible approach would be trying to decode all the codewords $\left\{\mathbf{x}_k\right\}_{k=1}^K$, thus treating the channel as a MAC with $K$ users. It is well-known (see for example~\cite{Cover}) that at high SNR, the achievable rates as dictated by the capacity region of the (Gaussian) MAC channel are essentially a ``zero-sum" game (up to a power gain), i.e., time sharing is nearly optimal. In particular, for such a channel, if all users are working at the same rate, the symmetric rate scales like \begin{align}
\frac{1}{2K}\log \text{SNR}.\nonumber
\end{align}
Since the decoder is only interested in one of the codewords, it seems wasteful to decode all of the interferers as well. For this reason, it is desirable to \emph{align} all interferers to one codeword, as was first noticed in~\cite{Bresler}. After alignment is performed, the receiver only has to decode two codewords: the desired codeword $\mathbf{x}_1$, and the aligned interference codeword.

A linear code, as defined in (\ref{linearity}), facilitates the task of aligning the $K-1$ interfering codewords into one codeword. Specifically, if all interfering codewords $\left\{\mathbf{x}_k\right\}_{k=2}^K$ are taken from the same linear code $\mathcal{C}$, and the channel gains $\left\{h_k\right\}_{k=2}^K$ associated with the interfering codewords are all integers, we have
\begin{align}
\left[\sum_{k=2}^K h_k \mathbf{x}_k\right]^*={\mathbf{x}_{\text{IF}}}\in\mathcal{C},\nonumber
\end{align}
and therefore the received vector can be reduced modulo the interval $\BI$ to yield
\begin{align}
\mathbf{y}^*&=\left[h_1\mathbf{x}_1+\left[\sum_{k=2}^K h_k \mathbf{x}_k\right]^*+\mathbf{z}\right]^*\nonumber\\
&=\left[h_1\mathbf{x}_1+{\mathbf{x}_{\text{IF}}}+\mathbf{z}\right]^*.
\label{inter2MAC}
\end{align}
Since $\mathbf{x}_1$ and ${\mathbf{x}_{\text{IF}}}$ are both members of the same linear codebook, the equivalent channel in (\ref{inter2MAC}) satisfies the conditions of Theorem~\ref{capacityTheorem}, and we can find an achievable symmetric rate for it.

At this point it is worth noting the advantage of joint decoding over successive decoding. A successive decoding procedure, as used in~\cite{VeryStrong} and~\cite{LayeredSymmetric}, can decode both codewords only if a very strong interference condition is satisfied; that is, one of the codewords can be treated as noise while decoding the other codeword. For a wide range of values of $h_1$, successive decoding does not allow for for positive transmission rates. The result of the previous section provides an achievable rate region that is greater than zero for a much wider range of values of $h_1$.

We next give a formal definition for the Gaussian interference channel, and then use the results of the previous section in order to derive achievable rates for certain classes of interference channels.

\subsection{The $K$-user Gaussian interference channel}
\label{KintChannel}

The $K$-user Gaussian interference channel consists of $K$ pairs of transmitters and receivers, where each transmitter $k$ tries to convey one message $w_k$ out of a set of $2^{nR_k}$ possible messages to its corresponding receiver. Specifically, the signal observed by receiver $j$ is
\begin{align}
Y_j=h_{jj}X_j+\sum_{k=1,k\neq j}^K h_{jk}X_k+Z_j,\nonumber
\end{align}
where $h_{jk}$ is the channel gain from transmitter $k$ to receiver $j$, and $Z_j$ is the Gaussian noise present at receiver $j$. All transmitters and receivers have perfect knowledge of all channel gains. We assume that the Gaussian noise at each receiver is i.i.d. with zero mean and variance $1/\text{SNR}$ and that the noises at different receivers are statistically independent. We assume all transmitters are subject to the same power constraint
\begin{align}
\frac{1}{n}\mathbb{E}\left[\|\mathbf{x}_k\|^2\right]\leq 1.\nonumber
\end{align}
Each transmitter $k$ has an encoding function
\begin{align}
f_k: \left\{1,\ldots,2^{n R_k}\right\}\rightarrow \mathbb{R}^n,\nonumber
\end{align}
such that the signal transmitted by user $k$ during $n$ channel uses is
\begin{align}
\mathbf{x}_k=f_k(w_k).\nonumber
\end{align}
Receiver $j$ recovers the message using a decoding function
\begin{align}
g_j: \mathbb{R}^n\rightarrow \left\{1,\ldots,2^{nR_j}\right\}.\nonumber
\end{align}
Let
\begin{align}
\hat{w}_j=g_j(\mathbf{y}_j)\nonumber
\end{align}
be the estimate receiver $j$ produces for the message transmitted by transmitter $j$. We define the error probability as the probability that at least one of the receivers did not decode its intended message correctly
\begin{align}
\bar{P}_{e,\text{IF}}=\mathbb{E}\left[\Pr\left(\left\{\hat{w}_1,\ldots,\hat{w}_K\right\}\neq\left\{w_1,\ldots,w_K\right\}\right)\right],\nonumber
\end{align}
where the expectation here is over a uniform distribution on the messages.

We say that a rate-tuple $\left\{R_1,\ldots,R_K\right\}$ is achievable if there exists a set of encoding and decoding functions such that $\bar{P}_{e,\text{IF}}$ vanishes as $n$ goes to infinity, and that a symmetric rate $R_{\text{sym}}$ is achievable if the rate-tuple $\left\{R_{\text{sym}},\ldots,R_{\text{sym}}\right\}$ is achievable.

\subsection{The integer-interference channel}
\label{intSquareChannel}
We restrict attention to a special family of $K$-user Gaussian interference channels which we refer to as the integer-interference channel. In this family, all the channel gains corresponding to interferers are integers, i.e., for all $j\neq k$,
\begin{align}
h_{jk}=a_{jk}\in \mathbb{Z}.\nonumber
\end{align}
The following theorem establishes an achievable symmetric rate for the integer-interference channel.
\begin{theorem}
\label{intSquareTheorem}
For the $K$-user integer-interference channel, the following symmetric rate is achievable
\begin{align}
&R_{\text{sym}}<\max_{p\in\bigcap_{j=1}^K\mathcal{P}'(h_{jj})} \min_{j\in\{1,\ldots,K\}} \min\nonumber\\
&\bigg\{-\frac{1}{2}\log\left(\frac{1}{p^2}+\sqrt{\frac{2\pi/3}{\text{SNR}}}+\frac{1}{p}e^{-\frac{3\text{SNR}}{2p^2}\delta^2(p,h_{jj})}+2e^{-\frac{3\text{SNR}}{8}}\right),\nonumber\\
&-\log\left(\frac{1}{p}+\sqrt{\frac{2\pi/3}{\delta^2(p,h_{jj})\text{SNR}}}+2e^{-\frac{3\text{SNR}}{8}}\right)\bigg\},
\label{intSquareRsym}
\end{align}
where $\delta(\cdot,\cdot)$ is defined in (\ref{deltaDef}), and $\mathcal{P}'(\cdot)$ is defined in (\ref{thDef}).
\end{theorem}
\vspace{2mm}
\begin{proof}
We begin by recalling the encoding and decoding procedures.

\underline{\emph{Encoding}:}
\noindent Generate a linear ensemble of codebooks of rate $R_{\text{sym}}$ over $\mathbb{Z}_p$ as described in Section~\ref{sec:MAC}, where $p$ is taken as the maximizing value in (\ref{intSquareRsym}). Choose the codebook $\mathcal{C}$ which achieves the smallest average error probability, $\bar{P}_{e,\text{IF}}$.
Each user encodes its message using this codebook.

\vspace{2mm}
\underline{\emph{Decoding}:}
Each receiver first reduces its observation modulo the interval $\BI$. The equivalent channel receiver $j$ sees is therefore
\begin{align}
{\mathbf{y}_j^*}
&=\left[h_{jj}\mathbf{x}_j+\sum_{k=1,k\neq j}^K a_{jk}\mathbf{x}_k+\mathbf{z}_j\right]^*\nonumber\\
&=\left[h_{jj}\mathbf{x}_j+{\mathbf{x}_{\text{IF},j}}+\mathbf{z}_j\right]^*,\nonumber
\end{align}
where
\begin{align}
{\mathbf{x}_{\text{IF},j}}=\left[\sum_{k=1,k\neq j}^K a_{jk}\mathbf{x}_k\right]^*.\nonumber
\end{align}
The linearity of the codebook implies that ${\mathbf{x}_{\text{IF},j}}\in\mathcal{C}$.

From Theorem~\ref{capacityTheorem} we know that $\mathbf{x}_j$ and ${\mathbf{x}_{\text{IF},j}}$ can be decoded reliably (by receiver $j$) as long as the symmetric rate $R_{\text{sym}}$ satisfies
\begin{align}
(a) \nonumber \\ R_{\text{sym}}&<-\frac{1}{2}\log\left(\frac{1}{p^2}+\sqrt{\frac{2\pi/3}{\text{SNR}}}+\frac{1}{p}e^{-\frac{3\text{SNR}}{2p^2}\delta^2(h_{jj},p)}+2e^{-\frac{3\text{SNR}}{8}}\right)
\label{Rate_j1}\\
(b) \nonumber  \\
R_{\text{sym}}&<-\log\left(\frac{1}{p}+\sqrt{\frac{2\pi/3}{\delta^2(p,h_{jj})\text{SNR}}}+2e^{-\frac{3\text{SNR}}{8}}\right),
\label{R_j2}
\end{align}
and
\begin{align}
p\in\mathcal{P}'(h_{jj}).
\label{Pcondition}
\end{align}
The codebook $\mathcal{C}$ satisfies conditions (\ref{Rate_j1}), (\ref{R_j2}) and (\ref{Pcondition}) for every\footnote{The existence of a codebook $\mathcal{C}$ that is simultaneously good for all $K$ equivalent MAC channels is guaranteed for any finite number of users $K$.} $1\leq j\leq K$, and thus the theorem is proved.

\end{proof}

\subsection{Integer-interference channel: Degrees of freedom}
\label{subsec:DoF}

Theorem~\ref{intSquareTheorem} provides an achievable symmetric rate for the integer-interference channel that is valid for any SNR. We now show that in the limit where the SNR goes to infinity, the coding scheme achieves $K/2$ degrees of freedom, which is the upper bound established in~\cite{DoF_upperBound}. This shows that for the integer-interference channel, the proposed scheme is optimal in a DoF sense, and thus recovers the asymptotic results of~\cite{Etkin}.

In order to find the number of DoF the scheme achieves, we need the following theorem from the field of Diophantine approximations, which is due to Khinchin.

\begin{theorem}[Khinchin]
\label{KhinchinTheorem}

For almost every $\gamma\in\mathbb{R}$, the number of solutions to the inequality
\begin{align}
\big|\gamma-\frac{a}{l}\big|\leq\Phi(l)\nonumber
\end{align}
for $a\in\mathbb{Z}$ and $l\in\mathbb{N}$, is finite if the series
\begin{align}
\sum_{l=1}^{\infty}l\Phi(l)\nonumber
\end{align}
converges, and infinite if it diverges.
\end{theorem}

\vspace{2mm}

\begin{proof}
See, e.g.,~\cite{Diophantine}.
\end{proof}

\vspace{2mm}

Setting $\Phi(l)=l^{-2-\epsilon_1}$ in Theorem~\ref{KhinchinTheorem}, it follows that for any $\epsilon_1>0$, and almost every $\gamma\in\mathbb{R}$, there exist an integer $l^*(\gamma,\epsilon_1)$ for which there are no solutions to the inequality
\begin{align}
l\cdot\bigg|\gamma-\frac{\lfloor l\gamma\rceil}{l}\bigg|\leq l^{-1-\epsilon_1},\label{deltaIneq1}
\end{align}
in the range $l>l^*(\gamma,\epsilon_1)$. Moreover, for such $\gamma$ there exist a constant $c_1>0$ for which
\begin{align}
l\cdot\bigg|\gamma-\frac{\lfloor l\gamma\rceil}{l}\bigg|\geq c_1,\label{deltaIneq2}
\end{align}
for every $l\leq l^*(\gamma,\epsilon_1)$. Combining~(\ref{deltaIneq1}) with~(\ref{deltaIneq2}), it follows that for almost every $\gamma\in\mathbb{R}$ there exist a positive integer $\tilde{l}^*(\gamma,\epsilon_1)$, such that if $p>\tilde{l}^*(\gamma,\epsilon_1)$ there are no solutions to the inequality
\begin{align}
l\cdot\bigg|\gamma-\frac{\lfloor l\gamma\rceil}{l}\bigg|\leq p^{-1-\epsilon_1},\nonumber
\end{align}
in the range $l\in\mathbb{Z}_p\backslash\{0\}$.
Thus, for any $\epsilon_1>0$ and $p$ large enough we have
\begin{align}
\delta(p,\gamma)\geq p^{-1-\epsilon_1},
\label{gammaBound}
\end{align}
for almost every $\gamma\in\mathbb{R}$.

We now show that the ratio between the symmetric rate as given by Theorem~\ref{intSquareTheorem} and $\nicefrac{1}{4}\log\text{SNR}$ approaches $1$ for almost every set of direct channel gains when the $\text{SNR}$ tends to infinity, and hence the number of DoF is $K/2$.

Before giving a formal definition to the number of degrees of freedom, we need a few preliminary definitions. We define an interference channel code $\mathcal{C}'$ as a set of encoders $\left\{f_k\right\}_{k=1}^K$ and decoders $\left\{g_k\right\}_{k=1}^K$.
We define an interference channel coding scheme as a family of interference channel codes $\left\{\mathcal{C}'(\text{SNR})\right\}$, and define $\mathcal{R}'(\text{SNR})$ as the set of all rate-tuples that are achievable for the interference channel code $\mathcal{C}'(\text{SNR})$.

\vspace{2mm}

\begin{definition}
\label{DoFDef}

An interference channel coding scheme $\left\{\mathcal{C}'(\text{SNR})\right\}$ is said to achieve $d$ degrees of freedom if
\begin{align}
\limsup_{\text{SNR}\rightarrow\infty}\max_{R_1,\ldots,R_K\in\mathcal{R}'(\text{SNR})}\frac{\sum_{k=1}^K R_k}{\frac{1}{2}\log{\text{SNR}}}=d.\nonumber
\end{align}
\end{definition}

\vspace{2mm}

In order to show that for asymptotic (high) SNR, the ratio between the symmetric rate $R_{\text{sym}}$ given by~(\ref{intSquareRsym}) and $\nicefrac{1}{4}\log\text{SNR}$ approaches $1$, we set $p=\text{SNR}^{1/4-\epsilon_2}$ (where $\epsilon_2>0$ is chosen such that $p$ is a prime number).
From~(\ref{gammaBound}), we see that for this choice, for high enough SNR and almost every $\gamma\in\mathbb{R}$ we have
\begin{align}
\delta(p,\gamma)>\text{SNR}^{-1/4+\epsilon_1\epsilon_2+\epsilon_2-\epsilon_1/4}.
\label{DiophantineDelta}
\end{align}
We now use (\ref{DiophantineDelta}) in order to find lower bounds on the maximal achievable symmetric rate of Theorem~\ref{intSquareTheorem} for asymptotic SNR conditions.

The argument of the logarithm in (\ref{Rate_j1}) can be upper bounded (after some straightforward algebra) by
\begin{align}
\frac{1}{p^2}&+\sqrt{\frac{2\pi/3}{\text{SNR}}}+\frac{1}{p}e^{-\frac{3\text{SNR}}{2p^2}\delta^2(p,h_{jj})}+2e^{-\frac{3\text{SNR}}{8}}\nonumber\\
&<\text{SNR}^{-1/2+2\epsilon_2}+\sqrt{\frac{2\pi}{3}}\text{SNR}^{-1/2}\nonumber\\
&+2\text{SNR}^{-1/4+\epsilon_2}e^{-\frac{3}{2}\text{SNR}^{4\epsilon_2+2\epsilon_1\epsilon_2-\epsilon_1/2}}+2e^{-\frac{3\text{SNR}}{8}},
\label{asymptot1}
\end{align}
and the argument of the logarithm in (\ref{R_j2}) by
\begin{align}
\frac{1}{p}&+\sqrt{\frac{2\pi/3}{\delta^2(p,h_{jj})\text{SNR}}}+2e^{-\frac{3\text{SNR}}{8}}\nonumber\\
&<\text{SNR}^{-1/4+\epsilon_2}+\sqrt{\frac{2\pi}{3}}\text{SNR}^{-1/4-\epsilon_2-\epsilon_1\epsilon_2+\epsilon_1/4}+2e^{-\frac{3\text{SNR}}{8}}.
\label{asymptot2}
\end{align}
Taking $\epsilon_1$ and $\epsilon_2$ to zero such that
\begin{align}
4\epsilon_2+2\epsilon_1\epsilon_2-\epsilon_1/2>0,\nonumber
\end{align}
and taking $\text{SNR}$ to infinity we see that the r.h.s. of~(\ref{asymptot1}) and~(\ref{asymptot2}) are approximately $\text{SNR}^{-1/2}$ and $\text{SNR}^{-1/4}$ respectively. Since~(\ref{asymptot1}) and~(\ref{asymptot2}) hold for almost every $h_{jj}\in\mathbb{R}$, they simultaneously hold for almost every set of direct channel gains $\left\{h_{jj}\right\}_{j=1}^K$ as well. We thus conclude that the ratio between the symmetric rate from Theorem~\ref{intSquareTheorem} and $\nicefrac{1}{4}\log{\text{SNR}}$ approaches $1$ when the SNR tends to infinity, and therefore the number of DoF is $K/2$ for almost every set of direct channel gains.

\subsection{Example}

Consider the $5$-user integer-interference channel where the channel gains are the entries of the matrix
\begin{align}
H=\left(
    \begin{array}{ccccc}
      h & 1 & 2 & 3 & 4 \\
      5 & h & 3 & 6 & 7 \\
      2 & 11 & h & 1 & 3 \\
      3 & 7 & 6 & h & 9 \\
      11 & 2 & 6 & 4 & h \\
    \end{array}
  \right).
  \label{channelMat}
\end{align}
For this channel, we plot the achievable sum rate of our scheme (which is 5 times the symmetric rate $R_{\text{sym}}$), and for reference we also plot the sum rate a time sharing scheme would have achieved. One more curve we plot for reference is the curve
\begin{align}
\frac{K}{2}\frac{1}{2}\log(1+(1+h^2)\text{SNR}),
\label{DoFupperBound}
\end{align}
which corresponds to the sum rate that could have been achieved if the symmetric rate for a two-user Gaussian MAC with one linear code was the same as that of the same channel with two random codes, in other words, if $r_{\text{norm}}$ given in~(\ref{RnormEq}) were $1$. In the absence of explicit upper bounds for the $K$-user interference channel with finite SNR,~(\ref{DoFupperBound}) serves as a reasonable benchmark to the best performance one can expect to achieve, which is based on the known fact that the number of DoF the channel offers is $K/2$.

We consider two different values of $h$: $h=0.707$, and $h=\sqrt{2}/2$.\footnote{Our coding scheme would have the same performance for $h+m, \ m\in\mathbb{Z}$ as well; however the reference curves do change when adding integers to $h$.} The results are shown in Figure~\ref{intSqaureRateFig}. In Figure~\ref{intSqaureRateFig2} we plot the same curves for $h=0.24$ and $h=\sqrt{7}/11$.

\begin{figure}[htb]
\includegraphics[width=1 \columnwidth]{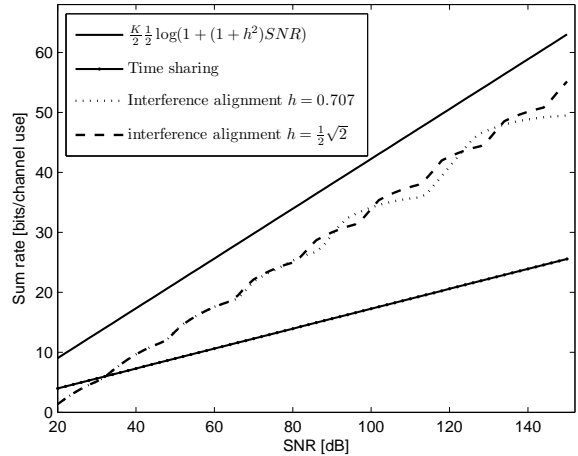}
\caption{The achievable sum rate for the integer-interference channel (\ref{channelMat}) for $h=0.707$ and $h=\sqrt{2}/2$.}
\label{intSqaureRateFig}
\end{figure}

\begin{figure}[htb]
\includegraphics[width=1 \columnwidth]{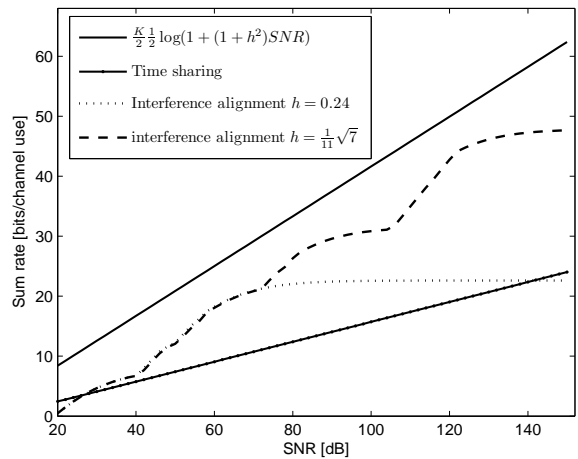}
\caption{The achievable sum rate for the integer-interference channel (\ref{channelMat}) for $h=0.24$ and $h=\sqrt{7}/11$.}
\label{intSqaureRateFig2}
\end{figure}

These examples show the advantages of interference alignment over time sharing for high enough SNR. For a larger number of users, interference alignment is preferable over time sharing for lower values of SNR. We note that in this paper we have not considered further optimizations for the achievable rate, such as combining it with time sharing of powers, or superposition (layered coding schemes), which may result in higher gains. Time sharing of powers, i.e. having all users transmit with more power for some of the time and remain silent for the rest of the time is important at low SNR values, since our coding scheme only achieves positive rates above some SNR threshold.

An important insight from Figures~\ref{intSqaureRateFig} and~\ref{intSqaureRateFig2} is the sensitivity of interference alignment to the channel gain $h$. Even though in each figure we have used values of $h$ that are very close, the performance of the scheme differs significantly when the SNR is very high. This sensitivity is the subject of the next section.

\section{Quantifying rationality}
\label{sec:rationales}

Previous results regarding interference alignment for the time-invariant $K$-user Gaussian interference channel were mainly focused on the degrees of freedom of the channel. In~\cite{Etkin} it is proved that the number of DoF of a $K$-user integer-interference channel are $K/2$ when the direct channel gains are irrational and is strictly smaller than $K/2$ where the channel gains are rational. This result is also supported by the results of~\cite{Khandany} where it is shown that the degrees of freedom of a $K$-user (not necessarily integer) interference channel is $K/2$, unless there exist some rational connections between the channel coefficients.

Clearly, from an engineering perspective, the dependence of the DoF on whether the direct channel gains are rational or irrational is very displeasing. It is therefore important to understand the effect of the channel gains being rational at finite SNR.

Theorem~\ref{intSquareTheorem} sheds light on this matter. Specifically, it quantifies the loss associated with a rational direct channel gain $h_{jj}$ in the symmetric rate as a function of how large the denominator of $h_{jj}$ is w.r.t. the SNR. Specifically, if $h_{jj}$ is a rational number of the form $h_{jj}=r/q$, the symmetric rate achieved by the coding scheme can never exceed $\log q$. This is evident from the presence of the factor $\delta(p,h_{jj})$ in (\ref{intSquareRsym}). For any $p>q$ we have $\delta(p,r/q)=0$. Thus, in order to get positive rates we must choose $p\leq q$ and the symmetric rate would be smaller than $\log q$ for any SNR.
For this reason, a small denominator $q$ limits the symmetric rate even for low values of SNR. However, a large value of $q$ limits performance only at high SNR. This phenomenon can be seen in Figure~\ref{intSqaureRateFig} where at a certain SNR point, the symmetric rate corresponding to $h=0.707$ (which is a rational number) saturates. In Figure~\ref{intSqaureRateFig2} this effect is even more pronounced for $h=0.24$ as the denominator of $h$ in this case is $q=25$, rather than $q=1000$ which is the case for $h=0.707$.
It is also seen from Figures~\ref{intSqaureRateFig} and~\ref{intSqaureRateFig2} that for low values of SNR, the symmetric rates corresponding to the irrational values of $h$ and their quantized rational versions are nearly indistinguishable.

Another question that arises from the results of \cite{Etkin} and \cite{Khandany} is how the rate behaves when the direct gains approach a rational number. Theorem~\ref{intSquareTheorem} provides an answer to this question as well. If $h_{jj}-r/q$ is small, then $\delta(h_{jj},p)$ would also be small for $p>q$, which would result in an effectively lower SNR. However, the function $\delta(p,h_{jj})$ is continuous in the second variable, and thus, letting $h_{jj}$ approach $r/q$ (where $q<p$) results in a \emph{continuous} decrease of the effective SNR.

\section{Non-integer interference channels}
\label{sec:nonInteger}
We have seen that for the integer-interference channel, the result of Theorem~\ref{capacityTheorem} was very useful for finding a new achievable rate region. The requirement that at each receiver all the channel gains corresponding to interferers are integers is necessary because the codebook we use is only closed under addition of integer-valued multiplications of codewords, which allows to align all interfering signals into one codeword.

Unfortunately, the integer-interference channel model does not capture the essence of the physical wireless medium. Under realistic statistical models for the interference channel, the probability of getting an integer-interference channel is clearly zero.

It is thus desirable to transform the original interference channel into an integer-interference channel by applying certain operations at the transmitters and the receivers.\footnote{We do not discuss the possibility of adding more antennas at the receivers or the transmitters which could also assist in the problem.} Specifically, it is necessary that at each receiver the ratios between all interference gains be rational, and then an appropriate scaling at each receiver can transform the channel into an integer-interference channel.

Assume that each receiver had scaled its observation such that one of the interference gains equals $1$. We would like to ``shape'' the other interfering gains seen at each receiver to be integers as well, using operations at the transmitters. It turns out that by using power back-off at each transmitter, i.e., each transmitter $k$ scales its codeword by a factor of $\alpha_k\leq1$ prior to transmission, it is possible to transform only $K-1$ (in addition to the channel gains that were equalized to $1$ by the receivers) of the total $K^2$ channel gains into integers. It follows that perfect alignment, i.e., alignment of all interferers at all receivers simultaneously, is not possible (by these methods) even for $K=3$ users.

One solution to this problem is performing partial alignment, as described in~\cite{Khandany}, which is suitable for almost every set of channel gains. This method roughly transforms the channel seen by each receiver into a MAC with a large number of inputs, where about half of the inputs correspond to the information transmitted by the desired transmitter, and the other half to interferences. At asymptotic (high) SNR conditions, it was shown in~\cite{Khandany} that this approach achieves $K/2$ degrees of freedom, i.e., each receiver is capable of decoding all the inputs corresponding to its intended messages in the presence of the interfering inputs.

An extension of the partial interference alignment approach proposed in~\cite{Khandany} to the non-asymptotic SNR regime, in the same manner we extended the results of~\cite{Etkin} for the integer-interference channel to non-asymptotic SNR, would require an extension of Theorem~\ref{capacityTheorem} to more than two users. Finding a non-trivial achievable symmetric rate for a Gaussian MAC with a general number of users, where all users are forced to use the same linear code, appears to be a rather difficult task. Moreover, the symmetric rate found in Theorem~\ref{capacityTheorem} depends on the ratio between the MAC channel coefficients. This dependence is expressed in our results through the factor $\delta(p,\gamma)$, which effectively reduces the SNR when it is small. The results for the two-user case suggest that for a larger number of users, the SNR loss caused by the ratio between the coefficients would significantly increase. If this is indeed true, partial interference alignment would not result in significant gains over time sharing, for the moderate to high SNR regime.

A different approach that is yet to be exhausted for the time-invariant interference channel, is using the time dimension in conjunction with power back-off in order to achieve alignment, and allow for joint decoding of the intended message and some function of the interferers at each receiver.

Such an approach can be thought of as the interference channel's dual of space-time codes, and we refer to it as ``power-time'' codes. An example of such a power-time code is given in Appendix~\ref{sec:powerTime}. The power-time code given in the example is suitable for the $3$-user interference channel (with arbitrary channel coefficients), and allows to achieve $9/8$ degrees of freedom for almost all channel gains. While it is already known from~\cite{Khandany} that the number of DoF offered by this channel is $3/2$, our ``power-time'' approach gives an explicit expression for the symmetric rate for finite SNR.

A third possible approach for general interference channels, is inspired by the recently proposed relaying strategy Compute-and-Forward~\cite{compAndForIeee}. Nazer and Gastpar have shown in~\cite{compAndForIeee} that when using nested lattice codes (see e.g.~\cite{nestedLattices}), it is not necessary to completely align all users. In~\cite{compAndForIeee}, a receiver sees a linear combination of signals transmitted by different users, and tries to decode an integer-valued linear combination of these signals. In order to do that, the decoder multiplies its observation by a scaling factor which plays the role of directing the channel gains towards the integer-valued coefficients of the linear combination it tries to decode. Since the channel gains  usually do not have rational connections between them, there would always be some residual error in the approximation of that integer-valued linear combination. By using a nested lattice scheme, this residual error can be made statistically independent of the transmitted signals, with a distribution that approaches that of an AWGN when the lattice dimension tends to infinity.

A similar technique can be applied to the interference channel as well. Each receiver will try to scale the gains of the interferers it sees towards some integer-valued vector, and then jointly decode this integer-valued linear combination, which is a point in the original lattice (which all transmitters have used for transmission), along with the desired lattice point transmitted by the user which tries to communicate with that receiver. The advantage of this approach is that it does not require perfect alignment. However, the residual noise caused by this imperfect alignment incurs some losses in the achieved rate, and will most likely accumulate when $K$ grows.

\section{Concluding remarks}
\label{sec:conclusions}

We studied the two-user Gaussian MAC where all users are restricted to use the same linear code, and derived a new coding theorem which provides an achievable symmetric rate for it. Some of the bounds in the derivation of the coding theorem are rather crude, and as a result our analysis may not be very tight. Deriving a tighter analysis would be an interesting avenue for future work, and of particular interest is an extension of our coding theorem to to the case where both users transmit codewords from the same nested lattice codebook, with a ``good'' coarse lattice, as opposed to the one dimensional lattice we used.
Nevertheless, we believe that the derived expression for the symmetric rate captures the essence of the problem of restricting both users to transmit codewords from the same linear code.
The new coding theorem was utilized for establishing a new achievable rate region for the integer-interference channel which is valid for any value of SNR.
For a wide range of channel parameters the new rate region is the best known, and in particular, in the limit of SNR approaching infinity it coincides with previously known asymptotic results. We discussed strategies which enable to apply our results to the general (non-integer) $K$-user interference channel, among which is the novel ``Power-Time'' codes approach.

\section*{Acknowledgement}
The authors would like to express their deep gratitude to Ayal Hitron and Ronen Dar for their helpful technical comments.

\begin{appendices}

\section{Proof of Lemma~\ref{moduloLemma}}
\label{proofOfLemmaMod}

\begin{proof}
Let
\begin{align}
\mathbf{v}'=\arg\min_{\mathbf{v}\in L\mathbb{Z}^n}
\left\|\mathbf{z}+\mathbf{u}+\mathbf{v}\right\|^2,\nonumber
\end{align}
and denote by $\mathcal{S}$ the set of indices for which the absolute value of $\mathbf{v}'$ is not greater than $L$, and by $\bar{\mathcal{S}}$ the set of indices for which it is greater than $L$, namely
\begin{align}
&\mathcal{S}=\left\{s:|v'(s)|\leq L\right\}\nonumber\\
&\bar{\mathcal{S}}=\left\{s:|v'(s)|> L\right\}.\nonumber
\end{align}

 Let $\mathbf{u}_{\mathcal{S}},\mathbf{v}_{\mathcal{S}},\mathbf{\tilde{v}}_{\mathcal{S}},\mathbf{z}_{\mathcal{S}},\mathbf{z}^*_{\mathcal{S}}$ be sub-vectors of $\mathbf{u},\mathbf{v},\mathbf{\tilde{v}},\mathbf{z},\mathbf{z}^*$ in the indices $\mathcal{S}$, and $\mathbf{u}_{\bar{\mathcal{S}}},\mathbf{v}_{\bar{\mathcal{S}}},\mathbf{\tilde{v}}_{\bar{\mathcal{S}}},\mathbf{z}_{\bar{\mathcal{S}}},\mathbf{z}^*_{\bar{\mathcal{S}}}$ be their sub-vectors in the indices $\bar{\mathcal{S}}$.
  It suffices to show the next two inequalities:
\begin{align}
&\left\|\mathbf{z}_{\mathcal{S}}\right\|^2-\min_{\tilde{\mathbf{v}}_{\mathcal{S}}\in L\mathbb{T}^{|\mathcal{S}|}}\left\|\mathbf{z}_{\mathcal{S}}+
\mathbf{u}_{\mathcal{S}}+\tilde{\mathbf{v}}_{\mathcal{S}}\right\|^2\geq\nonumber\\
&\left\|\mathbf{z}^*_{{\mathcal{S}}}
\right\|^2-\min_{\mathbf{{v}}_{\mathcal{S}}\in L\mathbb{T}^{|\mathcal{S}|}}\left\|
\mathbf{z}^*_{{\mathcal{S}}}+
\mathbf{u}_{\mathcal{S}}+\mathbf{v}_{\mathcal{S}}\right\|^2
\label{ineq1}
\end{align}
and
\begin{align}
&\left\|\mathbf{z}_{\bar{\mathcal{S}}}\right\|^2-\min_{\tilde{\mathbf{v}}_{\bar{\mathcal{S}}}\in L\mathbb{T}^{|\bar{\mathcal{S}}|}}\left\|\mathbf{z}_{\bar{\mathcal{S}}}+
\mathbf{u}_{\bar{\mathcal{S}}}+\tilde{\mathbf{v}}_{\bar{\mathcal{S}}}\right\|^2\geq\nonumber\\
&\left\|\mathbf{z}^*_{{\bar{\mathcal{S}}}}
\right\|^2-\min_{\mathbf{{v}}_{\bar{\mathcal{S}}}\in L\mathbb{T}^{|\bar{\mathcal{S}}|}}\left\|
\mathbf{z}^*_{{\bar{\mathcal{S}}}}+
\mathbf{u}_{\bar{\mathcal{S}}}+\mathbf{v}_{\bar{\mathcal{S}}}\right\|^2.
\label{ineq2}
\end{align}
The first inequality (\ref{ineq1}) is true since for every index $s\in\mathcal{S}$ we have
\begin{align}
\min_{\tilde{v}_s\in L\mathbb{T}}\left(z_s+u_s+\tilde{v}_s\right)^2=\min_{{v}_s\in L\mathbb{T}}\left(z^*_{s}+u_s+v_s\right)^2,\nonumber
\end{align}
and since $z_s^2\geq(z^*_s)^2$.

In order to see that (\ref{ineq2}), is true we note that for any index $\bar{s}$ (and in particular $\bar{s}\in\bar{\mathcal{S}}$) we have
\begin{align}
(z^*_{\bar{s}})^2-\min_{v_{\bar{s}}\in L\mathbb{T}}(z^*_{\bar{s}}+u_{\bar{s}}+v_{\bar{s}})^2\leq\left(\frac{L}{2}\right)^2.\nonumber
\end{align}
 On the other hand, for any $\bar{s}\in\bar{\mathcal{S}}$, the inequality $|z_{\bar{s}}|>L$ must hold, since otherwise the value of $v_{\bar{s}}\in L\mathbb{Z}$ minimizing $(z_{\bar{s}}+u_{\bar{s}}+v_{\bar{s}})^2$ cannot be greater than $L$. Moreover, since $|u_{\bar{s}}|\leq L/2$, it follows that
 \begin{align}
 \min_{\tilde{v}_{\bar{s}}\in L\mathbb{T}}(z_{\bar{s}}+u_{\bar{s}}+\tilde{v}_{\bar{s}})^2\leq\left(|z_{\bar{s}}|-\frac{L}{2}\right)^2.\nonumber
  \end{align}
  Finally, for every $\bar{s}\in\bar{\mathcal{S}}$ we can write
\begin{align}
(z_{\bar{s}})^2&-\min_{\tilde{v}_{\bar{s}}\in L\mathbb{T}}\left(z_{\bar{s}}+u_{\bar{s}}+\tilde{v}_{\bar{s}}\right)^2\geq(z_{\bar{s}})^2-\left(|z_{\bar{s}}|-\frac{L}{2}\right)^2\nonumber\\
&=\frac{L}{2}\left(2|z_{\bar{s}}|-\frac{L}{2}\right)>\left(\frac{L}{2}\right)^2\nonumber\\
&\geq(z^*_{\bar{s}})^2-\min_{v_{\bar{s}}\in L\mathbb{T}}\left(z^*_{\bar{s}}+u_{\bar{s}}+v_{\bar{s}}\right)^2,\nonumber
\end{align}
which establishes (\ref{ineq2}).
\end{proof}

\section{Derivation of the Error Probabilities for the Different Cases}
\label{dependencies}

We begin this section with four lemmas that are repeatedly used in the error probability derivations for the different cases of statistical dependencies. The proofs are rather cumbersome and are given in Appendix~\ref{subsec:proofs}.

\vspace{2mm}

\begin{lemma}
\label{lemmaExpectation}
Let $X$ be a random variable uniformly distributed over $\left[\frac{L}{p}\mathbb{Z}_p\right]^*$, and let $\Theta$ be some random variable statistically independent of $X$. Define the random variable $U$ by
\begin{align}
U=\left[X+\Theta\right]^*.\nonumber
\end{align}
The following inequality holds:
\begin{align}
\mathbb{E}\left[e^{-\frac{\text{SNR}}{8}U^2}\right]\leq\frac{1}{p}+\sqrt{\frac{2\pi/3}{\text{SNR}}}.\nonumber
\end{align}

\end{lemma}

\vspace{2mm}

\begin{lemma}
\label{lemmaExp2}

Let $X$ be a random variable uniformly distributed over $\left[\frac{L}{p}\mathbb{Z}_p\right]^*$, and $\Theta$ be some random variable statistically independent of $X$. Define

\begin{align}
U=\left[(f+\varepsilon)X+\Theta\right]^*,\nonumber
\end{align}
where $f$ is a constant in $\mathbb{Z}_p$, and $\varepsilon\in\mathbb{R}$.
For any $f\in\mathbb{Z}_p$
\begin{align}
\mathbb{E}\left[e^{-\frac{\text{SNR}}{8}U^2}\right]\leq\frac{1}{p}+\sqrt{\frac{2\pi/3}{\delta^2(p,\varepsilon)\text{SNR}}},\nonumber
\end{align}
where
\begin{align}
\delta(p,\varepsilon)=\min_{l\in\mathbb{Z}_p\backslash\{0\}}l\cdot\bigg|\varepsilon-\frac{\lfloor l\varepsilon\rceil}{l}\bigg|.
\label{deltaDef2}
\end{align}

\end{lemma}

\vspace{2mm}

\begin{lemma}
\label{lemmaExp3}

Let $X_1$, $X_2$ and $X_3$ be three i.i.d random variables uniformly distributed over $\left[\frac{L}{p}\mathbb{Z}_p\right]^*$, and $\gamma\in\mathbb{R}$ some arbitrary constant. Define
\begin{align}
U=\left[X_1+\gamma\left(X_2-X_3\right)\right]^*.
\label{Udef3}
\end{align}
Then
\begin{align}
\mathbb{E}\left[e^{-\frac{\text{SNR}}{8}U^2}\right]&\leq\frac{1}{p^2}+\sqrt{\frac{2\pi/3}{\text{SNR}}}+\frac{1}{p}e^{-\frac{3\text{SNR}}{2p^2}\delta^2(p,\gamma)},\nonumber
\end{align}
where $\delta(\cdot,\cdot)$ is defined in~(\ref{deltaDef2}).
\end{lemma}

\vspace{2mm}

\begin{lemma}
\label{lemmaExp4}

Let $X$ be a random variable uniformly distributed over $\left[\frac{L}{p}\mathbb{Z}_p\right]^*$, and $\varepsilon\in\mathbb{R}$ some arbitrary constant. Define
\begin{align}
U=\left[f X+\varepsilon X-\varepsilon \right[r X\left]^*\right]^*,
\label{Udef4}
\end{align}
where $f\in\mathbb{Z}_p\backslash\{0\}$ and $r\in\mathbb{Z}_p\backslash\{0,1\}$.
Then
\begin{align}
\mathbb{E}\left[e^{-\frac{\text{SNR}}{8}U^2}\right]\leq \frac{p-1}{p}+\frac{1}{p}e^{-\frac{3\text{SNR}}{2p^2}\left(\varepsilon_{\bmod [-\frac{1}{4},\frac{1}{4})}\right)^2}.\nonumber
\end{align}
\end{lemma}

\vspace{2mm}

We also state two properties that are extensively used (sometimes implicitly) throughout the calculations in this section:

\vspace{2mm}

\begin{property}
For any $X\in\mathbb{R}$ and $a\in\mathbb{Z}$
\begin{align}
\left[a\left[X\right]^*\right]^*=\left[aX\right]^*.\nonumber
\end{align}
\end{property}

\vspace{2mm}

\begin{property}
For any $\alpha>0$ and $\beta>0$
\begin{align}
\left[\alpha X\right]_{\bmod [-\frac{\beta}{2},\frac{\beta}{2})}=\alpha\left[ X\right]_{\bmod [-\frac{\beta}{2\alpha},\frac{\beta}{2\alpha})}.\nonumber
\end{align}
\end{property}

\vspace{2mm}

We now turn to the analysis of the different cases of statistical dependencies.

\subsection*{\underline{Case A}}

In this case we have
\begin{align}
U_A&=\left[X_1+\gamma X_2-X\bTh-\gamma X\bFo\right]^*,\nonumber
\end{align}
where $X_1$, $X_2$, $X\bTh$ and $X\bFo$ are four i.i.d. random variables uniformly distributed over $\left[\frac{L}{p}\mathbb{Z}_p\right]^*$.
Let
\begin{align}
\bar{X}_1=\left[X_1-X\bTh\right]^*,\nonumber
\end{align}
and note that $\bar{X}_1$ is uniformly distributed over $\left[\frac{L}{p}\mathbb{Z}_p\right]^*$, and is statistically independent of $X_2$ and $X\bFo$. We have
\begin{align}
U_A=\left[\bar{X}_1+\gamma(X_2-X\bFo)\right]^*.\nonumber
\end{align}
Applying Lemma~\ref{lemmaExp3} gives
\begin{align}
\mathbb{E}\left[e^{-\frac{\text{SNR}}{8}U^2}\right]\leq\frac{1}{p^2}+\sqrt{\frac{2\pi/3}{\text{SNR}}}+\frac{1}{p}e^{-\frac{3\text{SNR}}{2p^2}\delta^2(p,\gamma)}.
\label{OmegaAtmp}
\end{align}
Denote by $\Omega_A$ the value of $\Omega$ associated with case $A$. Substituting (\ref{OmegaAtmp}) into (\ref{OmegaEq}) we have
\begin{align}
\Omega_A<\frac{1}{p^2}+\sqrt{\frac{2\pi/3}{\text{SNR}}}+\frac{1}{p}e^{-\frac{3\text{SNR}}{2p^2}\delta^2(p,\gamma)}+2e^{-\frac{3\text{SNR}}{8}}.
\label{OmegaA}
\end{align}

\subsection*{\underline{Case B}}

In this case
\begin{align}
&U_B=\bigg[X_1+\gamma X_2-X\bTh-\gamma\left[aX_1+bX_2+cX\bTh\right]^*\bigg]^*\nonumber\\
&=\bigg[X_1+\lfloor\gamma\rceil X_2+\left(\gamma-\lfloor\gamma\rceil\right)X_2-X\bTh\nonumber\\
& \ -\lfloor\gamma\rceil\left[aX_1+bX_2+cX\bTh\right]^*\nonumber\\
& \ -\left(\gamma-\lfloor\gamma\rceil\right)\left[aX_1+bX_2+cX\bTh\right]^*\bigg]^*\nonumber\\
&=\bigg[(1-\lfloor\gamma\rceil a)X_1+(\lfloor\gamma\rceil-\lfloor\gamma\rceil b)X_2+(-1-\lfloor\gamma\rceil c)X\bTh\nonumber\\
& \ +(\gamma-\lfloor\gamma\rceil)\left(X_2-\big[aX_1+bX_2+cX\bTh\big]^*\right)\bigg]^*,\nonumber
\end{align}
where $X_1$, $X_2$ and $X\bTh$ are three i.i.d. random variables uniformly distributed over $\left[\frac{L}{p}\mathbb{Z}_p\right]^*$.
Define
\begin{align}
d_1=\left[1-\lfloor\gamma\rceil a\right]_{\bmod p},\nonumber
\end{align}
\begin{align}
d_2=\left[\lfloor\gamma\rceil-\lfloor\gamma\rceil b\right]_{\bmod p},\nonumber
\end{align}
\begin{align}
d_3=\left[-1-\lfloor\gamma\rceil c\right]_{\bmod p},\nonumber
\end{align}
and
\begin{align}
\varepsilon=\gamma-\lfloor\gamma\rceil\in[-\frac{1}{2},\frac{1}{2}).\nonumber
\end{align}
Using these notations we have
\begin{align}
U_B&=\bigg[\left[d_1 X_1+d_2 X_2 + d_3 X\bTh\right]^*\nonumber\\
& \ +\varepsilon\left(X_2-\left[aX_1+bX_2+cX\bTh\right]^*\right)\bigg]^*.
\label{sumElements}
\end{align}
We now show that if the vectors $\left[d_1 \ d_3\right]$ and $\left[a \ c\right]$ are linearly independent (over $\mathbb{Z}_p$), the first and the second elements in the sum are statistically independent. To that end we prove the following lemma.

\vspace{2mm}

\begin{lemma}
 Let $A$ be a full-rank deterministic matrix with dimensions $M\times M$ over $\mathbb{Z}_p$ where $p$ is a prime number, and $\mathbf{X}$ a vector with dimensions $M\times 1$ containing elements which are i.i.d. random variables uniformly distributed over $\mathbb{Z}_p$.

 Let $\mathbf{\bar{X}} =A\mathbf{X}$ (with operations over $\mathbb{Z}_p$). Then the elements of $\mathbf{\bar{X}}$ are also i.i.d. random variables uniformly distributed over $\mathbb{Z}_p$.
\label{lemmaIso}
\end{lemma}

\vspace{2mm}

\begin{proof}
Since $A$ is full-rank, for any vector $\mathbf{\bar{x}}$ there exists one and only one vector $\mathbf{x}$ that satisfies $\mathbf{\bar{x}} =A\mathbf{x}$. Since the elements of $\mathbf{X}$ are i.i.d. and uniformly distributed, the vector $\mathbf{X}$ is uniformly distributed over $\mathbb{Z}_p^M$, which implies that
\begin{align}
\Pr\left(\mathbf{\bar{X}}=\mathbf{\bar{x}}\right)=\Pr\left(\mathbf{X}=\mathbf{x}\right)=\left(\frac{1}{p}\right)^M\nonumber
\end{align}
for every vector $\mathbf{\bar{x}}\in\mathbb{Z}_p^M$. This in turn implies that the elements of $\mathbf{\bar{X}}$ are i.i.d with uniform distribution over $\mathbb{Z}_p$.
\end{proof}

\vspace{2mm}

Let
\begin{align}
\left(
  \begin{array}{c}
    \bar{X}_1 \\
    \bar{X}_2 \\
    \bar{X}_3 \\
  \end{array}
\right)=\left(
          \begin{array}{ccc}
            d1 & d2 & d3 \\
            0 & 1 & 0 \\
            a & b & c \\
          \end{array}
        \right)\left(\begin{array}{c}
                 X_1 \\
                 X_2 \\
                 X\bTh
               \end{array}\right) \ {\bmod \BI}.
               \label{LinearEq}
\end{align}
With this notation~(\ref{sumElements}) can be rewritten as
\begin{align}
U_B&=\left[\bar{X}_1 +\varepsilon(\bar{X}_2-\bar{X}_3)\right]^*.\nonumber
\end{align}
We now have to distinguish between the case where $\left[d_1 \ d_3\right]$ and $\left[a \ c\right]$ are linearly independent which will be referred to as \emph{Case $B1$}, and the case where they are linearly dependent. The case where $\left[d_1 \ d_3\right]$ and $\left[a \ c\right]$ are linearly dependent and $\left[a \ c\right]\neq\left[0 \ 0\right]$ will be called \emph{Case $B2$}, and the case where $\left[a \ c\right]=\left[0 \ 0\right]$ will be called \emph{Case $B3$}.

\vspace{2mm}

\emph{\underline{Case B1}} - If $\left[d_1 \ d_3\right]$ and $\left[a \ c\right]$ are linearly independent, the matrix in (\ref{LinearEq}) is full rank, and it follows from Lemma~\ref{lemmaIso} that $\left\{\bar{X}_1,\bar{X}_2,\bar{X}_3\right\}$ are each uniformly distributed over $\left[\frac{L}{p}\mathbb{Z}_p\right]^*$ and are statistically independent. In this case
\begin{align}
U_{B1}=\left[\bar{X}_1+\Theta_{B1}\right]^*,\nonumber
\end{align}
where
\begin{align}
\Theta_{B1}=\varepsilon(\bar{X}_2-\bar{X}_3)\nonumber
\end{align}
 is statistically independent of $\bar{X}_1$.

Applying Lemma~\ref{lemmaExpectation} gives
\begin{align}
\mathbb{E}\left[e^{-\frac{\text{SNR}}{8}U_{B1}^2}\right]\leq\frac{1}{p}+\sqrt{\frac{2\pi/3}{\text{SNR}}}.
\label{OmegaB1tmp}
\end{align}
Denote by $\Omega_{B1}$ the value of $\Omega$ associated with case $B1$. Substituting (\ref{OmegaB1tmp}) into (\ref{OmegaEq}) we have
\begin{align}
\Omega_{B1}<\frac{1}{p}+\sqrt{\frac{2\pi/3}{\text{SNR}}}+2e^{-\frac{3\text{SNR}}{8}}.
\label{OmegaB1}
\end{align}

\vspace{2mm}

\emph{\underline{Case B2}} - We now consider the case where $\left[d_1 \ d_3\right]$ and $\left[a \ c\right]$ are linearly dependent and $\left[a \ c\right]\neq\left[0 \ 0\right]$. In this case we have
\begin{align}
\left[d_1 X_1+d_3 X\bTh\right]^*=\left[r\left(a X_1+c X\bTh\right)\right]^*,\nonumber
\end{align}
for some $r\in\mathbb{Z}_p$. Let
\begin{align}
X_3=\left[a X_1 + c X\bTh\right]^*,\nonumber
\end{align}
and note that $X_3$ is uniformly distributed over $\left[\frac{L}{p}\mathbb{Z}_p\right]^*$ and is statistically independent of $X_2$. We can rewrite $U_{B2}$ as
\begin{align}
U_{B2}=\left[rX_3+d_2 X_2+\varepsilon\left(X_2-\left[X_3+b X_2\right]^*\right)\right]^*.\nonumber
\end{align}
Now let
\begin{align}
X_4=[X_3+b X_2]^*,\nonumber
\end{align}
and note that $X_4$ is statistically independent of $X_2$. Using this notation, we have
\begin{align}
U_{B2}&=\left[rX_4+\left(d_2-r b\right)X_2+\varepsilon\left(X_2-X_4\right)\right]^*\nonumber\\
&=\left[\left(f+\varepsilon\right)X_2+\left(r-\varepsilon\right)X_4\right]^*,
\label{VcaseB}
\end{align}
where
\begin{align}
f=\left[d_2-r b\right]_{\bmod p}.\nonumber
\end{align}
Let
\begin{align}
\Theta_{B2}=\left(r-\varepsilon\right)X_4,\nonumber
\end{align}
and note that $\Theta_{B2}$ is statistically independent of $X_2$.
Using Lemma~\ref{lemmaExp2} we have
\begin{align}
\mathbb{E}\left[e^{-\frac{\text{SNR}}{8}U_{B2}^2}\right]\leq\frac{1}{p}+\sqrt{\frac{2\pi/3}{\delta^2(p,\gamma)\text{SNR}}},
\label{OmegaB2tmp}
\end{align}
which means that
\begin{align}
\Omega_{B2}<\frac{1}{p}+\sqrt{\frac{2\pi/3}{\delta^2(p,\varepsilon)\text{SNR}}}+2e^{-\frac{3\text{SNR}}{8}}.
\label{OmegaB2tmp}
\end{align}
We further note that
\begin{align}
\delta(p,\gamma)&=\min_{l\in\mathbb{Z}_p\backslash\{0\}}l\cdot\bigg|\gamma-\frac{\lfloor l\gamma\rceil}{l}\bigg|\nonumber\\
&=\min_{l\in\mathbb{Z}_p\backslash\{0\}}l\cdot\bigg|\varepsilon-\frac{\lfloor l\varepsilon\rceil}{l}\bigg|=\delta(p,\varepsilon).\nonumber
\end{align}
Thus, (\ref{OmegaB2tmp}) is equivalent to
\begin{align}
\Omega_{B2}<\frac{1}{p}+\sqrt{\frac{2\pi/3}{\delta^2(p,\gamma)\text{SNR}}}+2e^{-\frac{3\text{SNR}}{8}}.
\label{OmegaB2}
\end{align}

\vspace{2mm}

\emph{\underline{Case $B3$}} - Since $a=c=0$ it follows that $d_1=1$ and $d_3=[-1]_{\bmod p}=p-1$. Thus,~(\ref{sumElements}) can be written as
\begin{align}
U_{B3}&=\bigg[X_1+d_2 X_2 + (p-1) X\bTh+\varepsilon\left(X_2-\left[bX_2\right]^*\right)\bigg]^*.\nonumber
\end{align}
Letting
\begin{align}
\Theta_{B3}&=\bigg[d_2 X_2 + (p-1) X\bTh+\varepsilon\left(X_2-\left[bX_2\right]^*\right)\bigg]^*,\nonumber
\end{align}
and using the fact that $X_1$ is uniformly distributed over $\left[\frac{L}{p}\mathbb{Z}_p\right]^*$ and is statistically independent of $\Theta_{B3}$, Lemma~\ref{lemmaExpectation} can be applied, which yields
\begin{align}
\mathbb{E}\left[e^{-\frac{\text{SNR}}{8}U_{B3}^2}\right]\leq\frac{1}{p}+\sqrt{\frac{2\pi/3}{\text{SNR}}}.
\label{OmegaB3tmp}
\end{align}
Denote by $\Omega_{B3}$ the value of $\Omega$ associated with case $B3$. Substituting (\ref{OmegaB3tmp}) into (\ref{OmegaEq}) gives
\begin{align}
\Omega_{B3}<\frac{1}{p}+\sqrt{\frac{2\pi/3}{\text{SNR}}}+2e^{-\frac{3\text{SNR}}{8}}.
\label{OmegaB3}
\end{align}
Since $\delta^2(p,\gamma)<1$, combining~(\ref{OmegaB1}) with~(\ref{OmegaB2}) and~(\ref{OmegaB3}) yields
\begin{align}
\Omega_B<\frac{1}{p}+\sqrt{\frac{2\pi/3}{\delta^2(p,\varepsilon)\text{SNR}}}+2e^{-\frac{3\text{SNR}}{8}},
\label{OmegaBtmp}
\end{align}
for all possible values of $a$, $b$, and $c$.

\subsection*{\underline{Case C}}

In this case
\begin{align}
U&=\big[X_1+\gamma X_2-\left[aX_1+bX_2+cX\bFo\right]^*-\gamma X\bFo\big]^*\nonumber\\
&=\left[(1-a)X_1+(\gamma-b)X_2+(-\gamma-c)X\bFo\right]^*,
\label{thirdCase}
\end{align}
where $X_1$, $X_2$ and $X\bFo$ are three i.i.d. random variables uniformly distributed over $\left[\frac{L}{p}\mathbb{Z}_p\right]^*$.

We distinguish between the case where $a\neq 1$, which we refer to as \emph{Case $C1$}, and the case where $a=1$, which we refer to as \emph{Case $C2$}.
\vspace{2mm}

\emph{\underline{Case $C1$}} - Since $a\neq 1$ the random variable $\left[(1-a)X_1\right]^*$ is uniformly distributed over $\left[\frac{L}{p}\mathbb{Z}_p\right]^*$. We further define the random variable
\begin{align}
\Theta_{C1}=(\gamma-b)X_2+(-\gamma-c)X\bFo,\nonumber
\end{align}
which is statistically independent of $\left[(1-a)X_1\right]^*$.
Applying Lemma~\ref{lemmaExpectation} yields
\begin{align}
\mathbb{E}\left[e^{-\frac{\text{SNR}}{8}U_{C1}^2}\right]\leq\frac{1}{p}+\sqrt{\frac{2\pi/3}{\text{SNR}}}.
\label{OmegaC1tmp}
\end{align}
Denote by $\Omega_{C1}$ the value of $\Omega$ associated with case $C1$. Substituting (\ref{OmegaC1tmp}) into (\ref{OmegaEq}) we have
\begin{align}
\Omega_{C1}<\frac{1}{p}+\sqrt{\frac{2\pi/3}{\text{SNR}}}+2e^{-\frac{3\text{SNR}}{8}}.
\label{OmegaC1}
\end{align}

\vspace{2mm}

\emph{\underline{Case $C2$}} - Since $a=1$,~(\ref{thirdCase}) can be rewritten as
\begin{align}
U_{C2}&=\left[(\gamma-b)X_2+(-\gamma-c)X\bFo\right]^*\nonumber\\
&=\left[(\varepsilon+\lfloor\gamma\rceil-b)X_2+(-\varepsilon-\lfloor\gamma\rceil-c)X\bFo\right]^*\nonumber\\
&=\left[(f+\varepsilon)X_2+(r-\varepsilon)X\bFo\right]^*
\label{Vc}
\end{align}
where
\begin{align}
\varepsilon=\gamma-\lfloor\gamma\rceil,\nonumber
\end{align}
\begin{align}
f=\left[\lfloor\gamma\rceil-b\right]_{\bmod p},\nonumber
\end{align}
and
\begin{align}
r=\left[-\lfloor\gamma\rceil-c\right]_{\bmod p}.\nonumber
\end{align}
Letting
\begin{align}
\Theta_{C2}=(r-\varepsilon)X\bFo,\nonumber
\end{align}
and applying Lemma~\ref{lemmaExp2} gives
\begin{align}
\Omega_{C2}<\frac{1}{p}+\sqrt{\frac{2\pi/3}{\delta^2(p,\gamma)\text{SNR}}}+2e^{-\frac{3\text{SNR}}{8}}.
\label{OmegaC2}
\end{align}
Combining (\ref{OmegaC1}) with (\ref{OmegaC2}), and using the fact that $\delta(p,\gamma)<1$, yields
\begin{align}
\Omega_C<\frac{1}{p}+\sqrt{\frac{2\pi/3}{\delta^2(p,\gamma)\text{SNR}}}+2e^{-\frac{3\text{SNR}}{8}},
\label{OmegaC}
\end{align}
for all possible values of $a$, $b$, and $c$.

\subsection*{\underline{Case D}}

In this case we have
\begin{align}
U_D&=\big[X_1+\gamma X_2-aX_1-bX_2-\gamma \left[cX_1+dX_2\right]^*\big]^*\nonumber\\
&=\bigg[(1-a-\lfloor\gamma\rceil c)X_1+(\lfloor\gamma\rceil-b-\lfloor\gamma\rceil d)X_2\nonumber\\
& \ \ \ \  +\varepsilon\left(X_2-\left[cX_1+dX_2\right]^*\right)\bigg]^*,
\label{Case4}
\end{align}
where $X_1$ and $X_2$ are two i.i.d. random variables uniformly distributed over $\left[\frac{L}{p}\mathbb{Z}_p\right]^*$, and $\varepsilon=\gamma-\lfloor\gamma\rceil$ as before. Further, letting
\begin{align}
d_1=\left[1-a-\lfloor\gamma\rceil c\right]_{\bmod p},\nonumber
\end{align}
and
\begin{align}
d_2=\left[\lfloor\gamma\rceil-b-\lfloor\gamma\rceil d\right]_{\bmod p},\nonumber
\end{align}
we have
\begin{align}
U_{D}&=\left[d_1 X_1+d_2 X_2+\varepsilon\left(X_2-\left[cX_1+dX_2\right]^*\right)\right]^*.\nonumber
\end{align}
This is the most complicated case in terms of the number of different combinations of $a$, $b$, $c$ and $d$ we have to consider. \emph{Case $D1$} corresponds to $c\neq 0$, \emph{Case $D2$} to $\{c=0,a\neq 1\}$, \emph{Case $D3$} to $\{c=0,a=1,d=1\}$  and finally the case where $\{c=0,a=1,d\geq 2\}$ will be referred to as \emph{Case $D4$}.
The case  where $\{c=0,a=1,d=0\}$ does not have to be considered as in this case $\mathbf{w}\bFo=0$, and hence the message vectors $\mathbf{w}\bTh$ and $\mathbf{w}\bFo$ are linearly dependent. As we recall, the decoder in our scheme does not consider such pairs of message vectors.

We denote by $\Omega_{Di}, \ i=1,\ldots,4$, the value of $\Omega$ associated with Case $Di$.

\vspace{2mm}

\emph{\underline{Case $D1$}} - Define

\begin{align}
X_3=[c X_1+d X_2]^*,\nonumber
\end{align}
which is statistically independent of $X_2$ as $c\neq0$.
Now
\begin{align}
U_{D1}&=\left[d_1 X_1+d_2 X_2+\varepsilon(X_2-X_3)\right]^*\nonumber\\
&=\left[(d_1 c^{-1})cX_1+d_2 X_2+\varepsilon(X_2-X_3)\right]^*\nonumber\\
&=\bigg[(d_1 c^{-1})(cX_1+d X_2)\nonumber\\
& \ +(d_2- d_1 c^{-1}d)X_2+\varepsilon(X_2-X_3)\bigg]^*,
\label{UdTmp}
\end{align}
where $c^{-1}\in\mathbb{Z}_p$ is the inverse element of $c$ in the field $\mathbb{Z}_p$. Let
\begin{align}
r=\left[d_1 c^{-1}\right]_{\bmod p},\nonumber
\end{align}
and
\begin{align}
f=\left[d_2- d_1 c^{-1}d\right]_{\bmod p}.\nonumber
\end{align}
With these notations, (\ref{UdTmp}) can be written as
\begin{align}
U_{D1}&=\left[r X_3+f X_2+\varepsilon(X_2-X_3)\right]^*\nonumber\\
&=\left[(f+\varepsilon) X_2+(r-\varepsilon) X_3\right]^*.
\label{Ud}
\end{align}
Letting
\begin{align}
\Theta_{D1}=(r-\varepsilon) X_3,\nonumber
\end{align}
and applying Lemma~\ref{lemmaExp2} gives
\begin{align}
\Omega_{D1}<\frac{1}{p}+\sqrt{\frac{2\pi/3}{\delta^2(p,\gamma)\text{SNR}}}+2e^{-\frac{3\text{SNR}}{8}}.
\label{OmegaD1}
\end{align}

\vspace{2mm}

\emph{\underline{Case $D2$}} - Since $c=0$, (\ref{Case4}) becomes
\begin{align}
U_{D2}&=\bigg[(1-a)X_1+d_2 X_2+\varepsilon(X_2-\left[dX_2\right]^*)\bigg]^*.\nonumber
\end{align}
We define
\begin{align}
\Theta_{D2}=d_2 X_2+\varepsilon(X_2-\left[dX_2\right]^*),\nonumber
\end{align}
which is statistically independent of $(1-a)X_1$. Since $a\neq 1$ the random variable $[(a-1)X_1]^*$ is uniformly distributed over $\left[\frac{L}{p}\mathbb{Z}_p\right]^*$. We can therefore apply Lemma~\ref{lemmaExpectation} and get
\begin{align}
\Omega_{D2}<\frac{1}{p}+\sqrt{\frac{2\pi/3}{\text{SNR}}}+2e^{-\frac{3\text{SNR}}{8}}.
\label{OmegaD2}
\end{align}

\vspace{2mm}

\emph{\underline{Case $D3$}} - Substituting $c=0$, $a=1$ and $d=1$ into (\ref{Case4}) gives

\begin{align}
U_{D3}=[-b X_2]^*.\nonumber
\end{align}
Applying Lemma~\ref{lemmaExpectation} with $\Theta=0$ gives
\begin{align}
\Omega_{D3}<\frac{1}{p}+\sqrt{\frac{2\pi/3}{\text{SNR}}}+2e^{-\frac{3\text{SNR}}{8}},
\label{OmegaD3}
\end{align}
for any $b\neq 0$. The case $b=0$ is not interesting because it implies $\mathbf{w}\bTh=\mathbf{w}_1$ and $\mathbf{w}\bFo=\mathbf{w}_2$, and an error does not occur.

\vspace{2mm}

\emph{\underline{Case $D4$}} -
We are left only with the case $a=1,c=0,d\geq2$ for which
\begin{align}
U_{D4}=\left[d_2 X_2+\varepsilon X_2-\varepsilon\left[d X_2\right]^*\right]^*.
\label{X2eq}
\end{align}
Since $X_2$ is uniformly distributed over $\left[\frac{L}{p}\mathbb{Z}_p\right]^*$, $d_2\in\mathbb{Z}_p\backslash\{0\}$ and $d\in\mathbb{Z}_p\backslash\{0,1\}$, we can apply Lemma~\ref{lemmaExp4} which gives
\begin{align}
\Omega_{D4}<\frac{p-1}{p}+\frac{1}{p}e^{-\frac{3\text{SNR}}{2p^2}\left(\varepsilon_{\bmod [-\frac{1}{4},\frac{1}{4})}\right)^2}+2e^{-\frac{3\text{SNR}}{8}}.
\label{OmegaD4}
\end{align}
Combining (\ref{OmegaD1}), (\ref{OmegaD2}), (\ref{OmegaD3}), (\ref{OmegaD4}) and the fact that $\delta(p,\gamma)<1$, we conclude that for all values of $a$, $b$, $c$, and $d$ that are considered by the decoder (except for $a=1,b=0,c=0,d=1$ which does not incur an error event), we have\footnote{We also used the fact that $\gamma_{\bmod [-\frac{1}{4},\frac{1}{4})}=\varepsilon_{\bmod [-\frac{1}{4},\frac{1}{4})}$.}
\begin{align}
\Omega_D<\max\bigg\{&\frac{1}{p}+\sqrt{\frac{2\pi/3}{\delta^2(p,\gamma)\text{SNR}}}+2e^{-\frac{3\text{SNR}}{8}},\nonumber\\
&\frac{p-1}{p}+\frac{1}{p}e^{-\frac{3\text{SNR}}{2p^2}\left(\gamma_{\bmod [-\frac{1}{4},\frac{1}{4})}\right)^2}+2e^{-\frac{3\text{SNR}}{8}}\bigg\}.
\label{OmegaD}
\end{align}


\subsection{Proofs of Lemmas~\ref{lemmaExpectation}, \ref{lemmaExp2}, \ref{lemmaExp3} and~\ref{lemmaExp4}.}
\label{subsec:proofs}

The aim of this subsection is to prove Lemmas~\ref{lemmaExpectation}, \ref{lemmaExp2}, \ref{lemmaExp3} and~\ref{lemmaExp4}. We begin by deriving some auxiliary lemmas that will be used.
\vspace{2mm}
\begin{lemma}
\label{lemaExpSeriesTmp}
\mbox{}
\begin{enumerate}
\item [(a)]\label{partA}
For any $\theta\in\left[-\frac{1}{2},\frac{1}{2}\right)$ and $\rho>0$
\begin{align}
\sum_{k=-\infty}^{\infty} e^{-\rho(k+\theta)^2}<e^{-\rho\theta^2}+\int_{-\infty}^{\infty}e^{-\rho x^2}dx.
\label{sumEq}
\end{align}
\item [(b)]\label{partB}
For any $\theta\in\mathbb{R}$ and $\rho>0$
\begin{align}
\sum_{k=-\infty}^{\infty} e^{-\rho(k+\theta)^2}<e^{-\rho(\theta-\lfloor \theta\rceil)^2}+\sqrt{\frac{\pi}{\rho}}\leq 1+\sqrt{\frac{\pi}{\rho}}.
\label{sumEq1}
\end{align}
\end{enumerate}
\end{lemma}
\begin{proof}
In order to prove part~(a) of the lemma, we write
\begin{align}
\sum_{k=-\infty}^{\infty} e^{-\rho(k+\theta)^2}&
= e^{-\rho\theta^2}+\sum_{k=-\infty}^{-1} e^{-\rho(k+\theta)^2}+\sum_{k=1}^{\infty} e^{-\rho(k+\theta)^2}\nonumber\\
&= e^{-\rho\theta^2}+\sum_{k=1}^{\infty} e^{-\rho(k-\theta)^2}+\sum_{k=1}^{\infty} e^{-\rho(k+\theta)^2}.
\label{Proof2}
\end{align}
We have
\begin{align}
\int_{\theta}^{\infty}e^{-\rho x^2}dx&>\sum_{k=1}^{\infty} \min_{x\in[k-1+\theta,k+\theta]}e^{-\rho x^2}\nonumber\\
&=\sum_{k=1}^{\infty} e^{-\rho (k+\theta)^2},
\label{int2sum1}
\end{align}
and
\begin{align}
\int_{-\infty}^{\theta}e^{-\rho x^2}dx=\int_{-\theta}^{\infty}e^{-\rho x^2}dx&>\sum_{k=1}^{\infty} \min_{x\in[k-1-\theta,k-\theta]}e^{-\rho x^2}\nonumber\\
&=\sum_{k=1}^{\infty} e^{-\rho (k-\theta)^2}.
\label{int2sum2}
\end{align}
Substituting (\ref{int2sum1}) and (\ref{int2sum2}) into (\ref{Proof2}) yields
\begin{align}
\sum_{k=-\infty}^{\infty} e^{-\rho(k+\theta)^2}&< e^{-\rho\theta^2}+\int_{-\infty}^{\theta}e^{-\rho x^2}dx+\int_{\theta}^{\infty}e^{-\rho x^2}dx\nonumber\\
&=e^{-\rho\theta^2}+\int_{-\infty}^{\infty}e^{-\rho x^2}dx,
\label{Proof3}
\end{align}
which establishes the first part of the lemma.

In order to prove part~(b), we have
\begin{align}
\sum_{k=-\infty}^{\infty} e^{-\rho(k+\theta)^2}=\sum_{k=-\infty}^{\infty} e^{-\rho(k+\lfloor \theta\rceil+(\theta-\lfloor \theta\rceil))^2}.
\label{Proof2a}
\end{align}
Letting $\tilde{\theta}=\theta-\lfloor \theta\rceil\in\left[-\frac{1}{2},\frac{1}{2}\right)$ and $\tilde{k}=k+\lceil \theta\rfloor$, we have
\begin{align}
&\sum_{k=-\infty}^{\infty} e^{-\rho(k+\theta)^2}=\sum_{\tilde{k}=-\infty}^{\infty} e^{-\rho(\tilde{k}+\tilde{\theta})^2}\nonumber\\
&<e^{-\rho\tilde{\theta}^2}+\int_{-\infty}^{\infty}e^{-\rho x^2}dx\nonumber\\
&=e^{-\rho\tilde{\theta}^2}+\sqrt{\frac{\pi}{\rho}}\int_{-\infty}^{\infty}\frac{1}{\sqrt{2\pi\frac{1}{2\rho}}} e^{-\frac{1}{2}\frac{x^2}{1/2\rho}}dx\nonumber\\
&=e^{-\rho\tilde{\theta}^2}+\sqrt{\frac{\pi}{\rho}},
\label{Proof2b}
\end{align}
where we have used part~(a) of the lemma for the first inequality.
\end{proof}

\begin{lemma}
\label{lemmaMod}
For any $\theta\in\mathbb{R}$ and a prime number $p$
\begin{align}
\left[\mathbb{Z}_p+\theta\right]_{\bmod [-\frac{p}{2},\frac{p}{2})}\equiv \mathbb{Z}_p-\frac{p-1}{2}+\tilde{\theta},
\label{LemmaModEq}
\end{align}
where $\tilde{\theta}=\theta-\lfloor \theta\rceil\in\left[-\frac{1}{2},\frac{1}{2}\right)$, and the notation $\equiv$ stands for equality between sets of points (constellations).
\end{lemma}

\vspace{2mm}

\begin{proof}

\begin{align}
\left[\mathbb{Z}_p+\theta\right]&_{\bmod [-\frac{p}{2},\frac{p}{2})}\equiv \left[\mathbb{Z}_p+\lfloor \theta\rceil+\theta-\lfloor \theta\rceil\right]_{\bmod [-\frac{p}{2},\frac{p}{2})}\nonumber\\
&\equiv \left[\left[\mathbb{Z}_p+\lfloor \theta\rceil\right]_{\bmod [-\frac{p}{2},\frac{p}{2})}+\tilde{\theta}\right]_{\bmod [-\frac{p}{2},\frac{p}{2})}\nonumber\\
&\equiv \left[\mathbb{Z}_p-\frac{p-1}{2}+\tilde{\theta}\right]_{\bmod [-\frac{p}{2},\frac{p}{2})}\nonumber\\
&\equiv \mathbb{Z}_p-\frac{p-1}{2}+\tilde{\theta}.
\label{LemmaModProof}
\end{align}
\end{proof}
We are now ready to prove Lemma~\ref{lemmaExpectation}.
\begin{proof}[Proof of Lemma~\ref{lemmaExpectation}]
\begin{align}
\mathbb{E}&\left[e^{-\frac{\text{SNR}}{8}U^2}\right]=\mathbb{E}\left[\mathbb{E}\left[e^{-\frac{\text{SNR}}{8}U^2}|\Theta=\theta\right]\right]\nonumber\\
&\leq\max_{\theta\in\mathbb{R}}\mathbb{E}\left[e^{-\frac{\text{SNR}}{8}U^2}|\Theta=\theta\right]\nonumber\\
&=\max_{\theta\in\mathbb{R}}\frac{1}{p}\sum_{k=0}^{p-1}e^{-\frac{\text{SNR}}{8}\left(\left[\frac{L}{p}k+\theta\right]_{\bmod [-\frac{L}{2},\frac{L}{2})}\right)^2}\nonumber\\
&=\max_{\theta\in\mathbb{R}}\frac{1}{p}\sum_{k=0}^{p-1}e^{-\frac{\text{SNR}L^2}{8p^2}\left(\left[k+\frac{p}{L}\theta\right]_{\bmod [-\frac{p}{2},\frac{p}{2})}\right)^2}\nonumber\\
&\leq\max_{{\theta}\in\mathbb{R}}\frac{1}{p}\sum_{\tilde{k}=-\frac{p-1}{2}}^{\frac{p-1}{2}}e^{-\frac{\text{SNR}L^2}{8p^2}\left(\tilde{k}+\frac{p}{L}\theta-\lfloor\frac{p}{L}\theta\rceil\right)^2}\label{step1}\\
&\leq\max_{\tilde{\theta}\in[-0.5,0.5)}\frac{1}{p}\sum_{\tilde{k}=-\frac{p-1}{2}}^{\frac{p-1}{2}}e^{-\frac{\text{SNR}L^2}{8p^2}\left(\tilde{k}+\tilde{\theta}\right)^2}\nonumber\\
&\leq\frac{1}{p}\left(1+\sqrt{\frac{\pi\cdot 8 p^2}{12\text{SNR}}}\right)\label{step2}\\
&=\frac{1}{p}+\sqrt{\frac{2\pi/3}{\text{SNR}}},\nonumber
\end{align}
where (\ref{step1}) follows from Lemma~\ref{lemmaMod}, and (\ref{step2}) follows from part~(b) of Lemma~\ref{lemaExpSeriesTmp}, and the fact that $L^2=12$.
\end{proof}

\vspace{2mm}

Before proceeding to the proof of Lemma~\ref{lemmaExp2} we need to derive two more simple results.

\vspace{2mm}

\begin{lemma}
\label{lemmaConstellation}
Let $\mathcal{D}$ be some constellation of finite cardinality $|\mathcal{D}|$, with minimum distance
\begin{align}
d_{\text{min}}=\min_{x_1,x_2\in\mathcal{D},x_1\neq x_2}|x_1-x_2|.\nonumber
\end{align}
Then
\begin{align}
\sum_{x\in\mathcal{D}}e^{-\rho x^2}<\sum_{k=-\infty}^{\infty}e^{-\rho (k\cdot d_{\text{min}}+\theta)^2},
\label{NotTight}
\end{align}
for some $\theta\in\mathbb{R}$.
\end{lemma}
\begin{proof}
Let us sort the points of $\mathcal{D}$ in ascending order by $$d_0-\Delta_{-M_\text{neg}}\leq\ldots \leq d_0-\Delta_{-1}\leq d_0\leq d_0+\Delta_{1}\ldots \leq d_0+\Delta_{M_{\text{pos}}},$$
where
\begin{align}
d_0=\arg\min_{x\in{\mathcal{D}}}|x|.\nonumber
\end{align}
We have
\begin{align}
\sum_{x\in\mathcal{D}}e^{-\rho x^2}&=\sum_{k=1}^{M_\text{neg}}e^{-\rho (d_0-\Delta_{-k})^2}+e^{-\rho d_0^2}+\sum_{k=1}^{M_\text{pos}}e^{-\rho (d_0+\Delta_{k})^2}\nonumber\\
&\leq \sum_{k=1}^{M_\text{neg}}e^{-\rho (d_0-k\cdot d_{\text{min}})^2}+\sum_{k=0}^{M_\text{pos}}e^{-\rho (d_0+k\cdot d_{\text{min}})^2}\nonumber\\
&\leq \sum_{k=-\infty}^{\infty}e^{-\rho (d_0+k\cdot d_{\text{min}})^2}.\nonumber
\end{align}
Setting $\theta=d_0$ the lemma is proved.
\end{proof}
\begin{remark}
We note that the bound~(\ref{NotTight}) is rather loose, and is one of the weakest links in the chain of bounds we use for obtaining an upper bound on the average pairwise error probability $\mathbb{E}(\tpe)$.
\end{remark}

\begin{lemma}
\label{lemmaDmin}
Let
\begin{align}
X\equiv\left[\frac{L}{p}\mathbb{Z}_p\right]^*\equiv\frac{L}{p}\left[\mathbb{Z}_p\right]_{\bmod [-\frac{p}{2},\frac{p}{2})},\nonumber
\end{align}
and let $\mathcal{D}\equiv\left[(f+\varepsilon)X+\theta\right]^*$ where $f\in\mathbb{Z}_p$, and $\varepsilon,\theta\in\mathbb{R}$ are arbitrary constants. The minimum distance in the constellation $\mathcal{D}$ is lower bounded by
\begin{align}
d_{\text{min}}\geq\frac{L}{p}\min_{l\in\mathbb{Z}_p\backslash\{0\}}l\cdot\bigg|\varepsilon-\frac{\lfloor l\varepsilon\rceil}{l}\bigg|.\nonumber
\end{align}
\end{lemma}
\begin{proof}
The distance between any pair of distinct constellation points can be written as
\begin{align}
&\big|\left[(f+\varepsilon)x_1+\theta\right]^*-\left[(f+\varepsilon)x_2+\theta\right]^*\big|\nonumber\\
&\geq\big|\left[\left[(f+\varepsilon)x_1+\theta\right]^*-\left[(f+\varepsilon)x_2+\theta\right]^*\right]^*\big|\nonumber\\
&=\big|\left[f(x_1-x_2)+\varepsilon(x_1-x_2)\right]^*\big|,
\label{proof6a}
\end{align}
where $x_1$ and $x_2$ are two distinct points in $\left[\frac{L}{p}\mathbb{Z}_p\right]^*$.
Letting
\begin{align}
\tilde{x}_1=\frac{p}{L}x_1\in\left[\mathbb{Z}_p\right]_{\bmod [-\frac{p}{2},\frac{p}{2})}\nonumber
\end{align}
and
\begin{align}
\tilde{x}_2=\frac{p}{L}x_2\in\left[\mathbb{Z}_p\right]_{\bmod [-\frac{p}{2},\frac{p}{2})},\nonumber
\end{align}
we can further bound (\ref{proof6a}) as
\begin{align}
&\bigg|\left[f(x_1-x_2)+\varepsilon(x_1-x_2)\right]^*\bigg|\nonumber\\
&=\frac{L}{p}\bigg|\left[f(\tilde{x}_1-\tilde{x}_2)+\varepsilon(\tilde{x}_1-\tilde{x}_2)\right]_{\bmod [-\frac{p}{2},\frac{p}{2})}\bigg|\nonumber\\
&\geq\frac{L}{p}\bigg|\left[f(\tilde{x}_1-\tilde{x}_2)+\varepsilon(\tilde{x}_1-\tilde{x}_2)\right]_{\bmod [-\frac{1}{2},\frac{1}{2})}\bigg|\nonumber\\
&=\frac{L}{p}\bigg|\left[\varepsilon(\tilde{x}_1-\tilde{x}_2)\right]_{\bmod [-\frac{1}{2},\frac{1}{2})}\bigg|\nonumber\\
&\geq\frac{L}{p}\min_{l\in\mathbb{Z}_p\backslash\{0\}}\big|l\varepsilon-\lfloor l\varepsilon\rceil\big|\label{minimization}\\
&\geq\frac{L}{p}\min_{l\in\mathbb{Z}_p\backslash\{0\}}l\cdot\bigg|\varepsilon-\frac{\lfloor l\varepsilon\rceil}{l}\bigg|,\nonumber
\end{align}
where inequality (\ref{minimization}) is true since $0<|\tilde{x}_1-\tilde{x}_2|\leq p-1$.
\end{proof}
Aided by Lemmas~\ref{lemmaConstellation} and~\ref{lemmaDmin}, we can now prove Lemma~\ref{lemmaExp2}.
\vspace{2mm}
\begin{proof}[Proof of Lemma~\ref{lemmaExp2}]
\begin{align}
&\mathbb{E}\left[e^{-\frac{\text{SNR}}{8}U^2}\right]=\mathbb{E}\left[\mathbb{E}\left[e^{-\frac{\text{SNR}}{8}U^2}|\Theta=\theta\right]\right]\nonumber\\
&\leq\max_{\theta\in\mathbb{R}}\mathbb{E}\left[e^{-\frac{\text{SNR}}{8}U^2}|\Theta=\theta\right]\nonumber\\
&=\max_{\theta\in\mathbb{R}}\frac{1}{p}\sum_{k=0}^{p-1}e^{-\frac{\text{SNR}}{8}\left(\left[(f+\varepsilon)\frac{L}{p}k+\theta\right]_{\bmod [-\frac{L}{2},\frac{L}{2})}\right)^2}\nonumber\\
&<\max_{\tilde{\theta}\in\mathbb{R}}\frac{1}{p}\sum_{k=-\infty}^{\infty}e^{-\frac{\text{SNR}}{8}\left(k\cdot d_{\text{min}}+\tilde{\theta}\right)^2}\label{proof7aInter}\\
&=\max_{\tilde{\theta}\in\mathbb{R}}\frac{1}{p}\sum_{k=-\infty}^{\infty}e^{-\frac{\text{SNR}\cdot d^2_{\text{min}}}{8}\left(k +\frac{\tilde{\theta}}{d_{\text{min}}}\right)^2}\nonumber\\
&<\frac{1}{p}\left(1+\sqrt{\frac{8\pi}{\text{SNR}\cdot d^2_{\text{min}}}}\right)\label{proof7bInter}\\
&\leq \frac{1}{p}\left(1+\sqrt{\frac{8\pi p^2}{\text{SNR}\cdot \delta^2(p,\varepsilon) L^2}}\right)\label{proof7cInter}\\
&=\frac{1}{p}+\sqrt{\frac{2\pi/3}{\delta^2(p,\varepsilon)\text{SNR}}},\nonumber
\end{align}
where (\ref{proof7aInter}) follows from Lemma~\ref{lemmaConstellation}, (\ref{proof7bInter}) follows from part (b) of Lemma~\ref{lemaExpSeriesTmp}, and (\ref{proof7cInter}) from Lemma~\ref{lemmaDmin}.
\end{proof}
\vspace{2mm}
We use a similar technique for the proof of Lemma~\ref{lemmaExp3}.
\vspace{2mm}
\begin{proof}[Proof of Lemma~\ref{lemmaExp3}]
Let
\begin{align}
\tilde{X}_i=\frac{p}{L}X_i, \ \text{for }i=1,2,3,\nonumber
\end{align}
such that $\tilde{X}_1$, $\tilde{X}_2$ and $\tilde{X}_3$ are three i.i.d. random variables uniformly distributed over $\left[\mathbb{Z}_p\right]_{\bmod [-p/2,p/2)}$. With this notation, (\ref{Udef3}) can be written as
\begin{align}
U=\frac{L}{p}\left[\tilde{X}_1+\gamma\left(\tilde{X}_2-\tilde{X}_3\right)\right]_{\bmod [-\frac{p}{2},\frac{p}{2})}.\nonumber
\end{align}
Further, let
\begin{align}
\Theta=\tilde{X}_2-\tilde{X}_3,\nonumber
\end{align}
and note that for $\theta\in [0,\pm 1,\ldots,\pm (p-1)]$ we have
\begin{align}
\Pr\left(\Theta=\theta\right)=\frac{p-|\theta|}{p^2}\leq\frac{1}{p}.
\label{PrTheta}
\end{align}
Now,
\begin{align}
&\mathbb{E}\left[e^{-\frac{\text{SNR}}{8}U^2}\right]=\mathbb{E}\left[\mathbb{E}\left[e^{-\frac{\text{SNR}}{8}U^2}|\Theta=\theta\right]\right]\nonumber\\
&=\sum_{\theta=-(p-1)}^{p-1}\Pr\left(\Theta=\theta\right)\mathbb{E}\left[e^{-\frac{\text{SNR}}{8}U^2}|\Theta=\theta\right].
\label{PrThetaEq}
\end{align}
For any value of $\theta$, we have
\begin{align}
\mathbb{E}&\left[e^{-\frac{\text{SNR}}{8}U^2}|\Theta=\theta\right]\nonumber\\
&=\mathbb{E}\left[e^{-\frac{\text{SNR}}{8}\frac{L^2}{p^2}\left(\left[\tilde{X}_1+\gamma\theta\right]_{\bmod [-\frac{p}{2},\frac{p}{2})}\right)^2}\right]\nonumber\\
&=\frac{1}{p}\sum_{k=-\frac{p-1}{2}}^{\frac{p-1}{2}}e^{-\frac{3\text{SNR}}{2p^2}\left(k+(\gamma\theta-\lfloor \gamma\theta\rceil)\right)^2}\label{ConstellationBound}\\
&\leq \frac{1}{p}\left(e^{-\frac{3\text{SNR}}{2p^2}\left(\gamma\theta-\lfloor \gamma\theta\rceil\right)^2}+\sqrt{\frac{p^2 2\pi/3}{\text{SNR}}}\right)\label{ExponentSeriesBound}\\
&=\sqrt{\frac{2\pi/3}{\text{SNR}}}+\frac{1}{p}e^{-\frac{3\text{SNR}}{2p^2}\left(\gamma\theta-\lfloor \gamma\theta\rceil\right)^2},
\label{ThetaBound}
\end{align}
where~(\ref{ConstellationBound}) follows from Lemma~\ref{lemmaMod} and~(\ref{ExponentSeriesBound}) follows from part (b) of Lemma~\ref{lemaExpSeriesTmp} .

We note that for $\theta=0$, (\ref{ThetaBound}) becomes
\begin{align}
\mathbb{E}&\left[e^{-\frac{\text{SNR}}{8}U^2}|\Theta=0\right]\leq\sqrt{\frac{2\pi/3}{\text{SNR}}}+\frac{1}{p},
\label{ThetaIsZero}
\end{align}
and for any $\theta\neq 0 $
\begin{align}
\mathbb{E}&\left[e^{-\frac{\text{SNR}}{8}U^2}|\Theta\neq 0\right]\leq\sqrt{\frac{2\pi/3}{\text{SNR}}}+\frac{1}{p}e^{-\frac{3\text{SNR}}{2p^2}\delta^2(p,\gamma)},
\label{ThetaNotZero}
\end{align}
which follows directly from the definition of $\delta(p,\gamma)$ and the fact that $|\theta|\in\mathbb{Z}_p\backslash\{0\}$.

Substituting~(\ref{PrTheta}),~(\ref{ThetaIsZero}) and~(\ref{ThetaNotZero}) into~(\ref{PrThetaEq}) yields
\begin{align}
\mathbb{E}&\left[e^{-\frac{\text{SNR}}{8}U^2}\right]
\leq\frac{1}{p^2}+\frac{1}{p}e^{-\frac{3\text{SNR}}{2p^2}\delta^2(p,\gamma)}+\sqrt{\frac{2\pi/3}{\text{SNR}}}.\nonumber
\end{align}
\end{proof}

\vspace{2mm}

\begin{proof}[Proof of Lemma~\ref{lemmaExp4}]

The first step in proving the lemma is showing that for any $r\in\mathbb{Z}_p\backslash\{0,1\}$, there exists at least one value of $x\in\left[\frac{L}{p}\mathbb{Z}_p\right]^*$ for which $|u|>\frac{L}{p}|\varepsilon_{\bmod[-\frac{1}{4},\frac{1}{4})}|$.

For any value of $x\in\left[\frac{L}{p}\mathbb{Z}_p\right]^*$, define $\tilde{x}=\frac{p}{L}x\in\left[\mathbb{Z}_p\right]_{\bmod[-\frac{p}{2},\frac{p}{2})}$. We have
\begin{align}
|u|&=\left|\left[f x+\varepsilon x-\varepsilon \right[r x\left]^*\right]^*\right|\nonumber\\
&=\frac{L}{p}\left|\left[f \tilde{x}+\varepsilon \tilde{x}-\varepsilon \left[r \tilde{x}\right]_{\bmod[-\frac{p}{2},\frac{p}{2})}\right]_{\bmod[-\frac{p}{2},\frac{p}{2})}\right|\nonumber\\
&\geq\frac{L}{p}\left|\left[f \tilde{x}+\varepsilon \tilde{x}-\varepsilon \left[r \tilde{x}\right]_{\bmod[-\frac{p}{2},\frac{p}{2})}\right]_{\bmod[-\frac{1}{2},\frac{1}{2})}\right|\nonumber\\
&=\frac{L}{p}\left|\left[\varepsilon \left(\tilde{x}-\left[r \tilde{x}\right]_{\bmod[-\frac{p}{2},\frac{p}{2})}\right)\right]_{\bmod[-\frac{1}{2},\frac{1}{2})}\right|.
\label{Umax}
\end{align}
We focus on the expression
\begin{align}
\tilde{x}-\left[r \tilde{x}\right]_{\bmod[-\frac{p}{2},\frac{p}{2})}.
\label{diffExpression}
\end{align}
We show that for any value of $2\leq r<p-1$ there exists a value of $\tilde{x}\in\left[\mathbb{Z}_p\right]_{\bmod[-\frac{p}{2},\frac{p}{2})}$ for which~(\ref{diffExpression}) equals $1$. In order to see that, we observe that the equation
\begin{align}
&\left[\tilde{x}-\left[r \tilde{x}\right]_{\bmod[-\frac{p}{2},\frac{p}{2})}\right]_{\bmod[-\frac{p}{2},\frac{p}{2})}\nonumber\\
& \ \ \ \ \ \ \ \ \ =\left[(1-r)\tilde{x}\right]_{\bmod[-\frac{p}{2},\frac{p}{2})}=1,
\label{diffExpression2}
\end{align}
has a (single) solution for any $r\in\mathbb{Z}_p\backslash\{1\}$. Further, since
\begin{align}
\tilde{x}-\left[r \tilde{x}\right]_{\bmod[-\frac{p}{2},\frac{p}{2})}\in[-(p-1),(p-1)],
\label{diffExpression3}
\end{align}
it can be deduced that for any $r\in\mathbb{Z}_p\backslash\{1\}$, there exists a (single) value of $\tilde{x}$ for which~(\ref{diffExpression}) equals either $1$ or $-(p-1)$.
Since~(\ref{diffExpression}) equals $-(p-1)$ only if
\begin{align}
\tilde{x}=-\left[r\tilde{x}\right]_{\bmod[-\frac{p}{2},\frac{p}{2})}=-\frac{p-1}{2},\nonumber
\end{align}
which is possible only for $r=p-1$, we conclude that for any $2\leq r<p-1$ there is a value of $x$, which we denote $x^1$, for which~(\ref{diffExpression}) equals $1$. Substituting $x^1$ into~(\ref{Umax}) yields
\begin{align}
|u|\geq\frac{L}{p}\left|\varepsilon_{\bmod[-\frac{1}{2},\frac{1}{2})}\right|\geq\frac{L}{p}\left|\varepsilon_{\bmod[-\frac{1}{4},\frac{1}{4})}\right|.
\label{UforR}
\end{align}
We are left with the case $r=p-1$, for which
\begin{align}
\tilde{x}-[r\tilde{x}]_{\bmod [-\frac{p}{2},\frac{p}{2})}=\tilde{x}-[-\tilde{x}]=2\tilde{x}.
\label{TwoTildex}
\end{align}
Substituting~(\ref{TwoTildex}) into~(\ref{Umax}) gives
\begin{align}
|u|\geq\frac{L}{p}\left|[2\tilde{x}\varepsilon]_{\bmod[-\frac{1}{2},\frac{1}{2})}\right|.\nonumber
\end{align}
It follows that for $r=p-1$ and $\tilde{x}=1$
\begin{align}
|u|\geq\frac{L}{p}\left|[2\varepsilon]_{\bmod[-\frac{1}{2},\frac{1}{2})}\right|\geq\frac{L}{p}\left|\varepsilon_{\bmod[-\frac{1}{4},\frac{1}{4})}\right|.
\label{UforPminusOne}
\end{align}
Combining~(\ref{UforR}) and~(\ref{UforPminusOne}), we conclude that for any $r\in\mathbb{Z}_p\backslash\{0,1\}$, there exists at least one value of $x$ for which
\begin{align}
|u|\geq\frac{L}{p}\left|\varepsilon_{\bmod[-\frac{1}{4},\frac{1}{4})}\right|.
\label{Umax2}
\end{align}
Now, since at least one of the equiprobable $p$ possible values of $U$ is bigger (in absolute value) than $\frac{L}{p}\left|\varepsilon_{\bmod[-\frac{1}{4},\frac{1}{4})}\right|$, $\mathbb{E}\left[e^{-\frac{\text{SNR}}{8}U^2}\right]$ can be upper bounded by
\begin{align}
\mathbb{E}\left[e^{-\frac{\text{SNR}}{8}U^2}\right]&=\frac{1}{p}\sum_{u}e^{-\frac{\text{SNR}}{8}U^2}\nonumber\\
&\leq \frac{p-1}{p}+\frac{1}{p}e^{-\frac{L^2\text{SNR}}{8p^2}\left(\varepsilon_{\bmod[-\frac{1}{4},\frac{1}{4})}\right)^2}\nonumber\\
&= \frac{p-1}{p}+\frac{1}{p}e^{-\frac{3\text{SNR}}{2p^2}\left(\varepsilon_{\bmod[-\frac{1}{4},\frac{1}{4})}\right)^2}.
\end{align}

\end{proof}

\section{An Example of a $3$-User Interference Channel Power-Time Code}
\label{sec:powerTime}
In this section we introduce a coding scheme that utilizes both power back-off at the transmitters, and the time domain, in order to allow for perfect interference alignment (with some loss in the number of DoF). We illustrate the scheme by an example which is useful for the general $3$-user interference channel and achieves $9/8$ degrees of freedom out of the $3/2$ DoF afforded by the channel for almost all channel gains (see~\cite{Khandany}).

We consider the channel
\begin{align}
H=\left(
    \begin{array}{ccc}
      h_{11} & h_{12} & h_{13} \\
      h_{21} & h_{22} & h_{23} \\
      h_{31} & h_{32} & h_{33} \\
    \end{array}
  \right).\nonumber
\end{align}
We use the channel $4n$ times, in order to transmit $3$ codewords of length $n$ by each user. We refer to $n$ consecutive channel uses as a frame. The actions taken by the transmitters and the receiver vary from frame to frame, as will be described in detail.
All transmitters and receivers use the same linear codebook $\mathcal{C}$ of rate $R_{\text{sym}}$ and length $n$ during all frames. The codeword transmitted by user $k$ at frame $t$ is denoted by $\mathbf{x}_{k,t}$.

We assume that $h_{13}\geq h_{12}$, $h_{22}\geq h_{23}$ and $h_{32}\geq h_{31}$. There is no loss of generality in this assumption, as the scheme we now describe can be easily modified for different ratios between the channels gains. For all frames, receiver $1$ scales its observation by $1/h_{12}$, receiver $2$ scales its observation by $1/h_{23}$ and receiver $3$ scales its observation by $1/h_{31}$ such that the equivalent channel is
\begin{align}
\tilde{H}=\left(
    \begin{array}{ccc}
      \tilde{h}_{11} & 1 & \tilde{h}_{13} \\
      \tilde{h}_{21} & \tilde{h}_{22} & 1 \\
      1 & \tilde{h}_{32} & \tilde{h}_{33} \\
    \end{array}
  \right),\nonumber
\end{align}
where $\tilde{h}_{1j}=h_{1j}/h_{12}$, $\tilde{h}_{2j}=h_{2j}/h_{23}$ and $\tilde{h}_{3j}=h_{3j}/h_{31}$.

We describe the operations taken by the transmitters and the receivers at each frame.
\vspace{2mm}
\begin{enumerate}
\item \emph{Frame $1$}:
User $1$ transmits the codeword $\mathbf{x}_{1,1}$, user $2$ transmits the codeword $\mathbf{x}_{2,1}$ and user $3$ transmits the codeword $\mathbf{x}_{3,1}$.

User $3$ scales its codeword by the factor $\alpha_3=1/\tilde{h}_{13}$, and all other transmitters do not scale their codewords. The equivalent channel is thus

\begin{align}
\tilde{H}_1=\left(
    \begin{array}{ccc}
      \tilde{h}_{11} & 1 & 1 \\
      \tilde{h}_{21} & \tilde{h}_{22} & 1/\tilde{h}_{13} \\
      1 & \tilde{h}_{32} & \tilde{h}_{33}/\tilde{h}_{13} \\
    \end{array}
  \right).\nonumber
\end{align}
Due to the perfect alignment at receiver $1$, it can decode $\mathbf{x}_{1,1}$, the codeword transmitted by user $1$. The other receivers cannot decode their codewords at this stage.

\item \emph{Frame $2$}:
User $1$ transmits the codeword $\mathbf{x}_{1,2}$, user $2$ transmits the codeword $\mathbf{x}_{2,2}$ and user $3$ transmits the codeword $\mathbf{x}_{3,2}$.

User $1$ scales its codeword by the factor $\alpha_1=1/\tilde{h}_{21}$, and all other transmitters do not scale their codewords. The equivalent channel is thus

\begin{align}
\tilde{H}_2=\left(
    \begin{array}{ccc}
      \tilde{h}_{11}/\tilde{h}_{21} & 1 & \tilde{h}_{13} \\
      1 & \tilde{h}_{22} & 1 \\
      1/\tilde{h}_{21} & \tilde{h}_{32} & \tilde{h}_{33} \\
    \end{array}
  \right).\nonumber
\end{align}
Due to the perfect alignment at receiver $2$, it can decode $\mathbf{x}_{2,2}$, the codeword transmitted by user $2$. The other receivers cannot decode their codewords at this stage.

\item \emph{Frame $3$}:
User $1$ transmits the codeword $\mathbf{x}_{1,3}$, user $2$ transmits the codeword $\mathbf{x}_{2,3}$ and user $3$ transmits the codeword $\mathbf{x}_{3,3}$.

User $2$ scales its codeword by the factor $\alpha_2=1/\tilde{h}_{32}$, and all other transmitters do not scale their codewords. The equivalent channel is thus

\begin{align}
\tilde{H}_3=\left(
    \begin{array}{ccc}
      \tilde{h}_{11} & 1/\tilde{h}_{32} & \tilde{h}_{13} \\
      \tilde{h}_{21} & \tilde{h}_{22}/\tilde{h}_{32} & 1 \\
      1 & 1 & \tilde{h}_{33} \\
    \end{array}
  \right).\nonumber
\end{align}
Due to the perfect alignment at receiver $3$, it can decode $\mathbf{x}_{3,3}$, the codeword transmitted by user $3$. The other receivers cannot decode their codewords at this stage.

\item \emph{Frame $4$}:
In this frame, each user repeats a codeword it has already transmitted in one of the previous frames.
Specifically, user $1$ transmits the codeword $\mathbf{x}_{1,1}$, user $2$ transmits the codeword $\mathbf{x}_{2,2}$ and user $3$ transmits the codeword $\mathbf{x}_{3,3}$.

None of the users scale their codewords, such that the equivalent channel is $\tilde{H}$.
Now, receiver $k$ observes the signal

\begin{align}
\mathbf{y}_{k,4}=\sum_{j=1}^3 \tilde{h}_{jk}\mathbf{x}_{j,j}+\mathbf{z}_{k},\nonumber
\end{align}
where $\mathbf{z}_{k}$ is the Gaussian noise present at receiver $k$.

Since receiver $k$ had already decoded the codeword $\mathbf{x}_{k,k}$ in the $k_{th}$ frame, it can subtract $\tilde{h}_{kk}\mathbf{x}_{k,k}$ from $\mathbf{y}_{k,4}$ which results in the equivalent two-user MAC channel
\begin{align}
\bar{\mathbf{y}}_{k,4}=\mathbf{y}_{k,4}-\tilde{h}_{kk}\mathbf{x}_{k,k}=\sum_{j=1,j\neq k}^3 \tilde{h}_{jk}\mathbf{x}_{j,j}+\mathbf{z}_{k}.\nonumber
\end{align}
User $k$ can now decode the two codewords transmitted by the other users during the fourth frame. For instance, in this step, user $1$ decodes the codewords $\mathbf{x}_{2,2}$ and $\mathbf{x}_{3,3}$.

Now that user $1$ has the side information $\mathbf{x}_{2,2}$ and $\mathbf{x}_{3,3}$, it can return to its observations from the second and the third frames. It can subtract from $\mathbf{y}_{1,2}$ the term $\mathbf{x}_{2,2}$, leaving it with a two-user MAC which allows it to decode $\mathbf{x}_{1,2}$. In the same manner, it can subtract the term $\tilde{h}_{13}\mathbf{x}_{3,3}$ from $\mathbf{y}_{1,3}$, leaving it with a two-user MAC which allows it to decode $\mathbf{x}_{1,3}$.

The same procedure is done by each one of the other decoders.
\end{enumerate}

The described power-time code results in $3$ different codewords, each with a rate that scales like $\nicefrac{1}{4}\log\text{SNR}$, that were decoded by each decoder. Taking into account the (symbol) rate of the power-time code which is $3/4$ (since the fourth channel use is ``wasted''), we get a sum rate that scales like $\frac{9}{8}\frac{1}{2}\log\text{SNR}$, which means that the number of DoF is $9/8$ .
More importantly, using Theorem~\ref{capacityTheorem}, we can find an achievable symmetric rate for this power-time code for any SNR. From our DoF analysis of~\ref{subsec:DoF}, we know that for almost any channel realization there exist a certain value of SNR from which the described power-time code outperforms time sharing.

\end{appendices}

\bibliographystyle{IEEEtran}
\bibliography{OrBib2}

\end{document}